\documentclass[hidelinks,onefignum,onetabnum]{siamart220329}
\usepackage{xcolor}
\usepackage{array}
\usepackage{tabularx}
\usepackage{braket}
\usepackage{graphicx}
\usepackage{geometry}
\usepackage{mathrsfs}
\usepackage{latexsym}
\usepackage{amsmath,amssymb,amsfonts}
\usepackage{mleftright}
\usepackage{mathtools}
\usepackage{hyperref}
\usepackage[numeric,initials]{amsrefs}
\usepackage{cite}
\usepackage{graphicx}
\usepackage{wrapfig}
\usepackage{fancybox}
\usepackage{bm}
\usepackage[all,pdftex,arc,curve,color,frame,poly, arrow]{xy}

\usepackage{mathdots}
\newtheorem{example}[theorem]{Example}

\usepackage{tikz-cd}
\usepackage{enumitem}

\usepackage{thmtools}

\usepackage{url}

\renewcommand{\epsilon}{\varepsilon}
\renewcommand{\phi}{\varphi}

\newcommand{\Sp}{\operatorname{Sp}}
\newcommand{\SL}{\operatorname{SL}}
\newcommand{\GL}{\operatorname{GL}}

\newcommand{\Or}{\operatorname{O}}

\newcommand{\U}{\operatorname{U}}
\newcommand{\SO}{\operatorname{SO}}

\newcommand{\CC}{\mathbb C}

\newcommand{\RR}{\mathbb R}
\newcommand{\PP}{\mathbb P}

\newcommand{\im}{\mathrm i}

\newcommand{\supp}{\operatorname{supp}}
\newcommand{\Stab}{\operatorname{Stab}}

\newcommand{\cS}{\mathcal S}

\newcommand{\cH}{\mathcal H}
\def\bw#1{{\textstyle\bigwedge^{\hspace{-.2em}#1}}}

\newcommand{\defi}[1]{\textsf{#1}}
\usepackage{color}

\title{Tensor decompositions with applications to LU and SLOCC equivalence of multipartite pure states}

\author{Luke Oeding\thanks{Department of Mathematics and Statistics, Auburn University, Auburn, AL (\email{\{oeding,yzt0060\}@auburn.edu}).} \and Ian Tan$^*$
}

\begin{document}

\maketitle
\begin{abstract}
We introduce a broad lemma, one consequence of which is the higher order singular value decomposition (HOSVD) of tensors (DeLathauwer et al. 2000). By an analogous application of the lemma, we find a complex orthogonal version of the HOSVD. Kraus's algorithm (Kraus 2010) used the HOSVD to compute normal forms of almost all $n$-qubit pure states under the action of the local unitary group. Taking advantage of the double cover $\SL_2(\CC)\times \SL_2(\CC)\to\SO_4({\CC})$, we produce similar algorithms (distinguished by the parity of $n$) that compute normal forms for almost all $n$-qubit pure states under the action of the SLOCC group.
\end{abstract}

\begin{keywords}
    	tensor decomposition, HOSVD, entanglement, quantum information
\end{keywords}

\begin{MSCcodes}
    15A69, 	81Pxx
\end{MSCcodes}

\section{Introduction}\label{sec:intro}
Orbit classification is ubiquitous in mathematics. A group $G$ acting on a set $\cS$ induces an equivalence relation separating elements in $\cS$ into equivalence classes, or \defi{orbits}. If the number of orbits is finite, we can hope to classify orbits by giving a complete list of representative elements, or \defi{normal forms}, from each orbit. Otherwise, the following possibilities are available:
\begin{itemize}
\item Give an effective algorithm that takes as input  $x\in\cS$ and computes the unique normal form of all points on its orbit.
    \item Exhibit a complete, non-redundant list of parameterized families of normal forms.
\end{itemize}
Usually a complete description of orbits is a lofty goal. In this case a natural first step is to find normal forms for almost all $x\in\cS$, i.e., for every $x$ in a full-measure subset of $\cS$. 

We shall be primarily interested in the tensor setting. Specifically, for $1\leq i\leq n$ let the subgroup $G_i\leq \GL_{d_i}(\CC)$ act on the $d_i$-dimensional complex vector space $V_i$. The group $G=G_1\times\dots\times G_n$ acts on $V_1\otimes\dots\otimes V_n$ by the natural representation 
\[
(g_1,\dots,g_n)\mapsto g_1\otimes\dots\otimes g_n,\quad \text{$g_i\in G_i$ for all $1\leq i\leq n$.}
\]

Examples of orbit classification problems for tensors arise in quantum information. In this field one studies information-processing tasks that can be carried out by performing operations and measurements on collections of qubits (which are the quantum analog of classical bits). The projective point $\Phi'\in \mathbb{P}(\CC^2)^{\otimes n}$ corresponding to the tensor $\Phi\in (\CC^2)^{\otimes n}$ describes the state of an \defi{$n$-qubit} system. Similarly, $\Phi'\in\mathbb{P}(\CC^3)^{\otimes n}$ represents an \defi{$n$-qutrit} state and in general $\Phi'\in\mathbb{P}(\CC^d)^{\otimes n}$ represents an \defi{$n$-qudit} state. Note that we will work with unnormalized representatives $\Phi\in (\CC^d)^{\otimes n}$ since results can be easily transferred to the projectivization. 

One classification problem of interest to us is for the action of the \defi{Local Unitary} (LU) group $\U_d^{\times n}$ on the $\RR$-vector space $(\CC^d)^{\otimes n}$. As the name suggests, two $n$-qudit states are LU equivalent if one is related to another via local unitary operations on the qubits. There is also the coarser classification under the SLOCC group of \defi{Stochastic Local Operations and Classical Communication}, that is, the action of the Cartesian product of special linear groups $\SL_d(\CC)^{\times n}$ on the $\CC$-vector space $(\CC^d)^{\otimes n}$. In this case, two $n$-qudit states are SLOCC equivalent if there is a nonzero probability of transforming one to the other by local operations and classical communication.

\subsection{Historical highlights}
Some highlights of orbit classification for the SLOCC and LU groups include the following:
\begin{enumerate}[leftmargin=*]
    \item[1.] \textit{Classification for $n=2$.}
    The two-qudit SLOCC (resp. LU) classification problem is equivalent to that of $\SL_d(\CC)\times\SL_d(\CC)$ (resp. $\U_d\times \U_d$) acting on the space of $d\times d$ matrices $\CC^{d\times d}$ by left and right multiplication. In the SLOCC case, this problem is solved by row and column reduction. If $M\in\CC^{d\times d}$ is invertible, the normal form of $M$ is $\sqrt[d]{\det (M)}\cdot I_d$. Otherwise, the normal form of $M$ is $\text{diag}(1,\dots,1,0,\dots,0)$ where the number of 1's that appear is equal to the rank of $M$. In the LU case, the problem is solved by the singular value decomposition of matrices, which gives real diagonal normal forms. This is also known as \defi{Schmidt decomposition} in the context of quantum information \cite[Thm.~2.7]{nielsen2002quantum}.
    \item[2.] \textit{LU classification for $n\geq 3$.}
    \begin{itemize} 
    \item[(a)]  3 qubits: This case was solved by the generalized Schmidt decomposition of Ac\'in et al. \cite{acin2000generalized}. As with singular value decomposition, this gives a single parameterized family of normal forms.
    \item[(b)] $n$ qudits: According to \cite[pp. 35-36]{ChtDjo:NormalFormsTensRanksPureStatesPureQubits} a complete classification of $4$-qubit states is unknown. Normal forms for \textit{almost all} $n$-qudit states can be found using the higher order singular value decomposition \cite{BKraus}.
    \end{itemize}
    \item[3.] \textit{SLOCC classification for $n\geq 3$. }
    \begin{itemize}
        \item[(a)] 3 qubits: This case is considered to be classical, but was re-introduced by Gelfand, Kapranov, and Zelevinski \cite[Ex.~4.5, p.~478]{GKZ} and famously introduced to the quantum community by  D\"ur, Vidal, and Cirac \cite{dur2000}. 
        \item[(b)] 4 qubits: Chterental and Djokovic classified the orbits for this case \cite{ChtDjo:NormalFormsTensRanksPureStatesPureQubits}. The first complete and irredundant classification was by Dietrich et al. \cite{dietrich2022classification}, whose methods followed Vinberg's method and work of Antonyan \cite{Antonyan, oeding2022translation}.         
        \item[(c)] 3 qutrits:  Most recently, Di Trani, de Graaf, and Marrani \cite{di2023classification} gave a classification of orbits over the real numbers, building on work of Nurmiev, who gave a classification of orbits over the complex numbers \cite{Nurmiev} using methods of Vinberg and Elasvili \cite{Vinberg-Elasvili}. We also note prior work of \cite{Thrall-Chanler} and \cite{Ng} which also considered this case.
        \item[(d)] After these cases the situation is considered widely open. In particular we know that the problem of classifying tensors \cite{belitskii2003complexity} and tensor diagrams \cite{turner2020finite} is in the ``wild'' category. However, this should not necessarily stop us from seeking classifications of orbits in small dimensions. We note that
        Dietrich's 4-qubit classification and Nurmiev's 3-qutrit classification rely on Vinberg's method and utilize connections to the exceptional Lie algebras $\mathfrak{e}_7$ and $\mathfrak{e}_6$ respectively. It is not clear how one should generalize, however Holweck and Oeding explored one possible direction in recent works \cite{HolweckOedingE8, holweck2023hyperdeterminant} that studied hyperdeterminants in these settings, and \cite{HolweckOeding22}, which attempted to generalize the Jordan decomposition to larger tensors.
    \end{itemize}
\end{enumerate}

\subsection{Our approach}The strategy of the present article is to provide a novel generalization to both the classical singular value decomposition (SVD) of matrices and higher order singular value decomposition (HOSVD) of tensors \cite{DeLathauwer_DeMoor_Vandewalle_2000}. The HOSVD is a tensor generalization of the matrix SVD. Significantly for us, the HOSVD can be used to compute LU normal forms for almost all tensors. On the other hand, Chterental and Djokovic \cite{ChtDjo:NormalFormsTensRanksPureStatesPureQubits} generalize the SVD with an analogous matrix factorization where the unitary group is replaced with the complex orthogonal group. Combining these ideas, we produce a complex orthogonal version of the HOSVD. The existence of a homomorphism $\SL_2(\CC)\times\SL_2(\CC)\to\SO_4(\CC)$ paired with orthogonal HOSVD leads to an algorithm that produces SLOCC normal forms for almost all $n$-qubit tensors $\Phi\in (\CC^2)^{\otimes n}$, just as Kraus's algorithm \cite{BKraus} achieved for the LU group. The situation is quite different depending on the parity of $n$.

\subsection{Wider context and applications}
Though the world of tensors and their applications is quite broad, for concreteness we have chosen to cast our results in the language of quantum information. We refer to \cite{LandsbergTensorBook} for general background on the algebra and geometry of tensors,  \cite{LandsbergQuantum, landsberg2019very, zhahir2022entanglement} for modern views on tensors in quantum information, and \cite{KoldaBader} for a classic overview on broader applications of tensor decompositions.

The problem of classifying orbits of tensors is closely related to tensor decompositions. Let $\U_{d_1}\times\dots\times\U_{d_n}$ act on $V=\CC^{d_1}\otimes\dots\otimes \CC^{d_n}$. If $\Phi\in V$ has normal form $\Omega\in V$, then we can write
\[
\Phi = (U_1\otimes\dots\otimes U_n)\Omega, \quad U_i\in \U_{d_i}\text{ for all }1\leq i\leq n.
\]
If we add the condition that $\Omega$ is ``all-orthogonal,'' then this equation expresses the HOSVD of tensors introduced by De Lathauwer et al. \cite{DeLathauwer_DeMoor_Vandewalle_2000}. The SVD of matrices is a special case of HOSVD, which is a special case of Tucker decomposition \cite{KoldaBader}. However, the HOSVD falls short of classifying unitary orbits since the so-called ``core tensor" $\Omega$ is not unique (see Section 4.3 in \cite{KoldaBader}). In \cref{sec:tpt} we discuss an algorithm introduced by Kraus \cite{BKraus} for finding unique $\Omega$ given general $\Phi$ when $V=\CC^2\otimes\dots\otimes \CC^2$. 

The SVD has numerous applications, including the solution to partial least squares problems, principal component analysis, latent variable learning for linear functions, and the classification of 2-qubit quantum states up to local unitary transformation. The HOSVD has also found numerous applications, including in brain science \cite{Deshpande2017}, data visualization \cite{navasca2015random}, genetics \cite{omberg2007tensor}, recommender systems \cite{frolov2017tensor}, and plant biodiversity \cite{bernardi2021high}, to name a few examples.

\subsection{Organization} Here is an overview of the organization of this article. We address our notational conventions in \cref{sec:notation}. In \cref{sec:OHOSVD} we discuss a generalized approach that allows us to obtain not only the well-known HOSVD as a consequence, but also the complex orthogonal HOSVD. After this we discuss normal forms for general qubits in \cref{sec:normal}. In particular, we describe the situation for the LU group in \cref{sec:tpt}, for the SLOCC group when $n$ is even in \cref{sec:sleven}, and for the SLOCC group when $n$ is odd in \cref{sec:odd}. 

\subsection{Notation and conventions}\label{sec:notation} We will write $\im$ for the imaginary unit and use the symbol ``$\leq$'' to indicate subgroup containment. The following matrices will appear throughout the text:
\[T:=\frac{1}{\sqrt{2}}\begin{pmatrix}
        1 & 0 & 0 & 1\\
        0 & \im & \im & 0\\
        0 & -1 & 1 & 0\\
        \im & 0 & 0 & -\im
    \end{pmatrix}\quad\text{and}\quad
    J:=
    \begin{pmatrix}
        0 & 1 \\
        -1 & 0
    \end{pmatrix}.
    \]
These matrices are related by the equation $T^\top T=J\otimes J$.

Let $n>1$ be a natural number. We write $\cH_n=(\CC^2)^{\otimes n}$ to denote the Hilbert space of (unnormalized) $n$-qubit state vectors. We follow the physics convention to represent a fixed basis of $\CC^2$ by $\ket 0, \ket{1}$. The induced basis vectors on the tensor product $\cH_n$ are written as
\[
\ket{\textbf{v}}=\ket {v_1 v_2 \dots v_n}=\ket{v_1}\otimes\ket{v_2}\otimes\dots\otimes\ket{v_n}
\]
for all $\textbf{v}=(v_1,\dots,v_n)\in\{0,1\}^n.$ Then the coordinates of a tensor $\Phi\in \cH_n$ in this basis are indexed by $n$-tuples $\textbf{v}\in\{0,1\}^n$ as in the equation
$\Phi = \sum_{\textbf{v}\in\{0,1\}^n} \Phi_{\textbf{v}} \ket{\textbf{v}}$, where $\Phi_{\textbf{v}}\in\CC$. 

For $1\leq i\leq n$ let $V_i$ be a $d_i$-dimensional complex vector space, let $\GL(V_i)$ denote the group of invertible linear operators on $V_i$, and let $G_i$ denote a group with a representation $G_i\to \GL(V_i)$. 
For any vector space $V$ let $V^*$ denote the dual vector space of linear functionals on $V$.
There are natural representations $ G_1\times\dots\times G_n\to \GL(V_1\otimes\dots\otimes V_n)$ and $G_1\times\dots\times G_n\to \GL(V_1\oplus\dots\oplus V_n)$  defined respectively by taking the Kronecker product and the direct sum of linear operators.
We write $G_1\otimes \dots\otimes G_n$ and $G_1\oplus \dots\oplus G_n$ as the respective images of these representations.

Given $\Phi\in V_1\otimes\dots\otimes V_n$, let $\Phi_{(i)}$ denote the $d_i\times \prod_{j\neq i} d_j$ matrix that is associated with the linear map
\begin{equation*}
(V_1\otimes\dots \otimes V_{i-1}\otimes V_{i+1}\otimes\dots\otimes V_n)^* \to V_i  
\end{equation*}
corresponding to $\Phi$ via a choice of bases. 
This matrix is called the mode-$i$ flattening in the literature.
The actions of $g=(g_1,\dots, g_n)\in G_1\times\dots\times G_n$ on $\Phi$ and on $\Phi_{(i)}$ are related by the equation
\begin{equation}\label{flatact}
(g.\Phi)_{(i)}= g_i\Phi_{(i)} (g_1\otimes\dots\otimes g_{i-1}\otimes g_{i+1}\otimes \dots\otimes g_n)^\top =:
g_i\Phi_{(i)}{\widehat{g_i}}^\top.
\end{equation}
Note the definition of $\widehat{g_i}$ in Eq.~\eqref{flatact} above.

\section{HOSVD and its complex orthogonal counterpart}\label{sec:OHOSVD}
In this section we cast the Higher Order Singular Value Decomposition (HOSVD) in a more general light. In \cref{lem:core} we discuss the existence of core elements associated with a suitable set of reduction maps $\pi_i$ (see \cref{def:reduction}). This approach allows us to combine the concepts of the complex orthogonal SVD \cite[Theorem~2.10]{ChtDjo:NormalFormsTensRanksPureStatesPureQubits} and the HOSVD to produce the (complex) orthogonal HOSVD.

\subsection{Reduction maps}\label{sec:reductions} 
Here we introduce some families of functions that satisfy an important property. We call these functions \textit{reduction maps}. They will play a key role in finding normal forms for various group actions.
\begin{definition}\label{def:reduction} Let $G_i$ be a group for $1\leq i\leq n$ and suppose $\cS$ is a $G_1\times\dots \times G_n$-set. We say that a function $\pi:\cS \to\cS_i$ to a $G_i$-set $\cS_i$ is a \defi{reduction map} if \[\pi((g_1,\dots,g_n).x)=g_i.\pi(x)\] for all $(g_1,\dots,g_n)\in G_1\times\dots\times G_n$ and $x\in\cS$.
\end{definition}
Many examples of reduction maps appear in the literature around SLOCC and LU equivalence. One reason for this is their ability to generate invariants. By the definition of a reduction map, any $G_i$-invariant function $f$ on $\cS_i$ pulls back to a $G_1\times\dots\times G_n$-invariant function $f\circ \pi$ on $\cS$. The reduction maps that will be useful for us are summarized in \cref{fig:maps}. Each map makes use of one of the following bilinear forms
\begin{itemize}
    \item{} [$\RR$-Hermitian] $(v,w)\mapsto v^* w$, where $(v,w)\in\CC^d\times\CC^d \cong \RR^{2d}\times \RR^{2d}$,
    \item{} [$\CC$-Orthogonal] $(v,w)\mapsto v^\top w$, where $(v,w)\in\CC^d\times\CC^d$,
    \item{} [$\CC$-Symplectic] $(v,w)\mapsto v^\top Jw$, where $(v,w)\in\CC^2\times \CC^2$.
\end{itemize}
The unitary group $\U_d$ consists of the invertible operators preserving the first bilinear form, the complex orthogonal group $\Or_d$ consists of the invertible operators preserving the second bilinear form, and the special linear group $\SL_2 =\Sp_2$ consists of the invertible operators preserving the third bilinear form. To see that the last claim is true, consider the equation
\begin{equation}\label{alt}
    A^\top JA= \det(A) J,\quad \forall A\in\GL_2.
\end{equation}

\begin{example}
Let $\cS=\CC^{d_1}\oplus \dots \oplus \CC^{d_n}$
be a direct sum of representations of groups $G_i\to\GL(\CC^{d_i})$. For each $1\leq i\leq n$ the projection $\pi_i\colon \cS\to \CC^{d_i}$ is a reduction map.
\end{example}
\begin{example}\label{ex:trout}
    (Reduced density matrix \cite{nielsen2002quantum}) Consider the natural action of $\U_{d_1}\times \U_{d_2}$ on $\cS = V\otimes W$, with $d_1=\dim V$ and $d_2=\dim W$. Let $\cS_1$ be the $\RR$-vector space of $d_1\times d_1$ Hermitian matrices considered as a $\U_{d_1}$-module by the conjugation action. Define $\pi\colon V\otimes W\to \cS_1$ by $\pi(\Phi)=\Phi_{(1)}\Phi_{(1)}^*$ for $\Phi\in V\otimes W$. By Eq.~\eqref{flatact}, if $(U_1,U_2)\in \U_{d_1}\times\U_{d_2}$ we have
    \[\pi((U_1,U_2).\Phi)=U_1\Phi_{(1)} U_2^\top(U_1\Phi_{(1)} U_2^\top )^*=U_1\Phi_{(1)}\Phi_{(1)}^*U_1^*=U_1.\pi (\Phi).\]
Thus $\pi$ is a reduction map. In this case, $\pi(\Phi)$ coincides with the reduced density matrix of $\rho=\Phi\Phi^*$ on the subsystem $V$, i.e. the partial trace $\pi(\Phi)=\text{tr}_W(\rho)$.
\end{example}

\begin{example}\label{ex:oo}
    Consider the natural action of $\Or_{d_1}\times \Or_{d_2}$ on $\cS = V\otimes W$, with $d_1=\dim V$ and $d_2=\dim W$. Let $\mathcal{S}_1\cong S^2 V$ be the space of $d_1\times d_1$ complex symmetric matrices considered as an $\Or_{d_1}$-module by the conjugation action. Define $\pi\colon V\otimes W\to \mathcal{S}_i$ by $\pi(\Phi)=\Phi_{(1)}\Phi_{(1)}^\top$ for $\Phi\in V\otimes W$. By \eqref{flatact}, if $(U_1,U_2)\in \U_{d_1}\times\U_{d_2}$ we have
    \[\pi((U_1,U_2).\Phi)=U_1\Phi_{(1)} U_2^\top(U_1\Phi_{(1)} U_2^\top )^\top=U_1\Phi_{(1)}\Phi_{(1)}^\top U_1^\top=U_1.\pi (\Phi).\]
    Thus $\pi$ is a reduction map. 
\end{example}

\begin{example}\label{ex:slt} 
    Consider the SLOCC module $\cS = \cH_n = V_1\otimes\dots\otimes V_n$ where each $V_i$ is a copy of $\CC^2$. If $n$ is odd, let $\mathcal{S}_i\cong S^2 V_i$ be the space of $2\times 2$ complex symmetric matrices. Otherwise if $n$ is even, let $\mathcal{S}_i\cong\bw{2} V_i$ be the space of $2\times 2$ complex skew-symmetric matrices. In either case $\mathcal{S}_i$ is an $\SL_2$-module by the action $A.M=AMA^\top$ for $A\in \SL_2$ and $M\in\cS_i$. For $1\leq i\leq n$ define $\pi_i:\cH_n\to\mathcal{S}_i$ by $\pi_i(\Phi)=\Phi_{(i)}J^{\otimes (n-1)}\Phi_{(i)}^\top$ for $\Phi\in\cH_n$. By Eq.~\eqref{flatact} and Eq.~\eqref{alt}, we have
    \begin{align*}
    \pi_i((A_1,\dots,A_{n}).\Phi)&=A_i\Phi_{(i)}\widehat{A_i}^{\top}J^{\otimes (n-1)}(A_i\Phi_{(i)}\widehat{A_i}^{\top})^\top \\
    &=A_i\Phi_{(i)}\widehat{A_i}^{\top}J^{\otimes (n-1)}\widehat{A_i} \Phi_{(i)}^\top A_i^\top\\
    &=A_i\Phi_{(i)}J^{\otimes (n-1)}\Phi_{(i)}^\top A_i^\top\\
    &= A_i.\pi_i(\Phi)
\end{align*}
where $A_1,\dots, A_n\in\SL_2$. Thus $\pi_i$ is a reduction map. 
Note that $J^{\otimes k}$ is symmetric when $k$ is even and skew-symmetric otherwise, which explains the two cases for $\cS_i$.

Since the determinant of $M\in\cS_i$ is an $\SL_2$ invariant, it pulls back to the SLOCC invariant $\det(\pi_i)$ which was used in \cite{MR3706585,Gour_2011}. When $n=5$, a calculation shows that these SLOCC invariants span the space of degree 4 invariants in $\CC[\cH_5]$.

Note that the space of $2\times 2$ skew-symmetric matrices is one-dimensional. That is, if $n$ is even, then $\pi_i(\Phi)$ is a scalar multiple of $J$ for all $\Phi\in\cH_n$. In this case the $\SL_2$-action on $\cS_i$, by Eq.~\eqref{alt}, is trivial, and hence the reduction map is not very useful.
\end{example}

\begin{table}
    \begin{center}
\renewcommand{\arraystretch}{1.5}
\begin{tabular}{l|l|l}
 reduction map & source and target & groups\\ 
 \hline
 $\pi:\Phi\mapsto\Phi_{(1)}\Phi_{(1)}^*$ & $\CC^{d}\otimes \CC^{e}\to \CC^{d}\otimes {\CC^{d}}^*$ & $\U_{d} \times\U_{e}$ ; $\U_d$ \\  
 
 $\pi:\Phi\mapsto\Phi_{(1)}\Phi_{(1)}^\top$ & $\CC^{d}\otimes \CC^{e}\to S^2\CC^{d}$ & $\Or_{d}\times\Or_{e}$ ; $\Or_d$ \\ 
 
 $\pi_i:\Phi\mapsto \Phi_{(i)}J^{\otimes (n-1)}\Phi_{(i)}^\top$ & $\cH_n\to S^2\CC^{2}\text{ or }\bw{2}\CC^{2}$ & $\SL_2^{\times n}$ ; $\SL_2$ \\  
 
 $\pi_{ij}:\Phi\mapsto T\Phi_{(ij)}J^{\otimes (n-2)}\Phi_{(ij)}^\top T^\top$ & $\cH_n\to S^2\CC^{4}\text{ or } \bw{2}\CC^{4}$ & $\SL_2^{\times n}$ ; $\SL_2\times \SL_2\to \SO_4$
\end{tabular}
\end{center}
    \caption{Reduction maps of  \cref{ex:trout,ex:oo,ex:slt,ex:li}.}
    \label{fig:maps}
\end{table}

\begin{example}\label{ex:li}
    Consider the SLOCC module $\cS = \cH_n= V_1\otimes \dots\otimes V_n$ where $n\geq 3$ and each $V_i$ is a copy of $\CC^2$. If $n$ is even, let $\cS_{ij} \cong S^2(V_i\otimes V_j)$ be the space of $4\times 4$ complex symmetric matrices. Otherwise if $n$ is odd, let $\cS_{ij}\cong\bw{2} (V_i\otimes V_j)$ be the space of $4\times 4$ complex skew-symmetric matrices. The $4\times 4$ unitary matrix $T$ used in \cite{Verstraete_Dehaene_De_Moor_Verschelde_2002, ChtDjo:NormalFormsTensRanksPureStatesPureQubits} provides the following isomorphism.
    \begin{equation}\label{iso}
    \begin{gathered}
    \xymatrix@R-1.5pc{
    \SL_2\otimes\SL_2 \ar[r] & \SO_4\\
    A\otimes B\ar@{|->}[r] & T(A\otimes B)T^*
}
\end{gathered}
\end{equation}
We consider $\cS_{ij}$ an $\SL_2\times\SL_2$-module by the representation
\begin{equation}\label{eq:symrep}
\begin{gathered}\xymatrix@R-1.5pc{
		\SL_2\times\SL_2 \ar[r]& {\SO_4} \ar[r]^\rho& {\GL(\cS_{ij})}
\\
(A,B)\ar@{|->}[r] & T(A\otimes B)T^* \ar@{|->}[r] &\rho(T(A\otimes B)T^*)
}\end{gathered}
\end{equation}
where $\rho$ is defined by the conjugation action. For each $1\leq i<j\leq n$ let $\Phi_{(ij)}$ be the $4\times 2^{n-2}$ matrix corresponding to the linear map
\[
(V_1\otimes\dots \otimes V_{i-1}\otimes V_{i+1}\otimes\dots\otimes V_{j-1}\otimes V_{j+1}\otimes\dots\otimes V_n)^* \to V_i\otimes V_j.
\]
Finally, define $\pi_{ij}:V_1\otimes\dots\otimes V_n\to \cS_{ij}$ 
 by $\pi_{ij}(\Phi)=T\Phi_{(ij)}J^{\otimes {(n-2)}}\Phi_{(ij)}^\top T^\top$. By Eq.~\eqref{flatact} and Eq.~\eqref{alt}, we have
 \begin{align*}
\pi_{ij}((A_1\otimes\dots\otimes A_n)\Phi) &=T(A_i\otimes A_j)\Phi_{(ij)}\widehat{A_{ij}}^{\top} J^{\otimes (n-2)}((A_i\otimes A_j)\Phi_{(ij)}\widehat{A_{ij}}^{\top})^\top T^\top\\
&=T(A_i\otimes A_j)\Phi_{(ij)}\widehat{A_{ij}}^{\top} J^{\otimes (n-2)}(\widehat{A_{ij}}^{\top})^\top \Phi_{(ij)}^\top (A_i\otimes A_j)^\top T^\top\\
&=T(A_i\otimes A_j)T^* (T\Phi_{(ij)} J^{\otimes (n-2)} \Phi_{(ij)}^\top T^\top ) {T^*}^\top (A_i\otimes A_j)^\top T^\top\\
&=T(A_i\otimes A_j)T^* (T\Phi_{(ij)} J^{\otimes (n-2)} \Phi_{(ij)}^\top T^\top ) ({T} (A_i\otimes A_j) T^*)^\top\\
&= (A_i,A_j).\pi_{ij}(\Phi)
\end{align*}
where $A_1,\dots,A_n\in\SL_2$ and $\widehat{A_{ij}}=\otimes_{k\notin\{i,j\}} A_k$.
Therefore $\pi_{ij}$ is a reduction map via an isomorphism $\SL_2^{\times n} \to (\SL_2\times \SL_2)\times \SL_2^{\times (n-2)}$ rearranging factors of the group:
\[
(A_1,\dots,A_n)\mapsto ((A_i,A_j),A_1,\dots,A_{i-1}, A_{i+1},\dots A_{j-1}, A_{j+1},\dots,A_n).
\]

This reduction map was considered by Li \cite{LiSLOCCclassification}. By the discussion above, if $\Phi,\Phi'\in\cH_n$ are in the same SLOCC orbit, then $\pi_{ij}(\Phi)$ and $\pi_{ij}(\Phi')$ are in the same $\SO_4$-orbit by conjugation. Li showed the weaker claim that if $\Phi,\Phi'\in\cH_n$ are in the same SLOCC orbit, then $\pi_{ij}(\Phi)$ and $\pi_{ij}(\Phi')$ have the same Jordan normal form.
\end{example}

\subsection{Core elements} \label{sec:core}
In the HOSVD there is a notion of ``core tensor'' which takes the place of the diagonal matrix in the SVD matrix factorization. Our generalization of this notion (defined in \cref{lem:core} below) does not rely on tensor properties or linearity. Therefore we call it a ``core element,'' though we primarily have tensors in mind.
\begin{definition}
    Suppose a group $G$ acts on a set $\cS$. A \defi{normal form function} for the action is a function $F\colon \cS\to \cS$ such that, for any $x,y\in \cS$
    \begin{itemize}
        \item $x$ and $F(x)$ are in the same $G$-orbit, and
        \item $F(x)=F(y)$ if $x$ and $y$ are in the same $G$-orbit.  
    \end{itemize}
\end{definition}
Of course, normal-form functions always exist by the axiom of choice. One may view the orbit classification problem for $G$ acting on $\cS$ as the problem of finding a ``nice'' or easily computable normal form function on $\cS$.

\Cref{lem:core} is a ``workhorse'' lemma. We use it and the associated algorithm (see the comments after the proof) throughout the rest of the paper. Specifically, \cref{lem:core} is used in the proofs of \cref{thm:HOSVD,thm:OHOSVD,thm:sleven,lem:preodd}.
\begin{lemma}[Existence and Uniqueness of Core Elements]\label{lem:core}
    Let $G_1,\dots,G_n$ be groups. Suppose $G_1\times\dots\times G_n$ acts on a set $\cS$ and for each $1\leq i\leq n$ there exists a reduction map $\pi_i\colon \cS\to \cS_i$ to a $G_i$-set $\cS_i$:
    \begin{equation}\label{flat}
        \pi_i((g_1,\dots,g_n).x)=g_i.\pi_i(x)\quad \text{for all $g_i\in G_i, x\in\cS$}.
    \end{equation}
    Fix a normal form function $F_i\colon \cS_i\to \cS_i$ for the $G_i$-action on $\cS_i$. Then for each $x\in \cS$ there exists a \defi{core element} $\omega\in \cS$, which is defined by the properties:
    \begin{itemize}
        \item $x=(g_1,\dots, g_n).\omega$ for some $g_i\in G_i$, and
        \item $\pi_i(\omega)=F_i(\pi_i(\omega))$ for all $1\leq i\leq n$.
    \end{itemize}
    Moreover, the core element is unique up to the action of $H_1\times\dots \times H_n\leq G_1\times\dots\times G_n$, where
    $H_i=\{g\in G_i:g.\pi_i(\omega)=\pi_i(\omega)\}$
    is the stabilizer subgroup of $\pi_i(\omega)$.
\end{lemma}
\begin{proof}
    Let $x\in\cS$ and for each $i$ write $N_i=F_i(\pi_i(x))$. By definition $N_i$ and $\pi_i(x)$ are in the same orbit, hence there exists $g_i\in G_i$ such that $\pi_i(x)=g_i.N_i$. Let $\omega=(g_1^{-1},\dots, g_n^{-1}).x$. Then by Eq.~\eqref{flat},
\[
\pi_i(\omega)=g_i^{-1}.\pi_i(x)=g_i^{-1}g_i.N_i=N_i.\]
    Since $x$ and $\omega$ are in the same orbit, by Eq.~\eqref{flat} so are $\pi_i(x)$ and $\pi_i(\omega)$. Then $N_i=F_i(\pi_i(\omega))$, hence $\omega$ is a core tensor for $x$. This proves existence.
    
    Now for uniqueness, suppose $\xi$ is another core element for $x$, i.e. $x=(h_1,\dots, h_n).\xi$ for some $h_i\in G_i$ and $\pi_i(\xi)=N_i$ for all $i$. Then by Eq.~\eqref{flat},
\[
h_i.N_i=h_i.\pi_i(\xi)=\pi_i(x)=g_i.\pi_i(\omega)=g_i.N_i
\]
    so $h_i^{-1}g_i$ fixes $N_i$ and $(h^{-1}_1g_1,\dots, h^{-1}_ng_n).\omega=\xi$. That is, $\xi$ is in the $H_1\times\dots\times H_n$-orbit of $\omega$.
\end{proof}

The proof of \cref{lem:core} also gives us an algorithm to compute the core element $\omega$, supposing that for each $1\leq i\leq n$ and $x_i\in\cS_i$ one already has a way to find $g_i\in G_i$ such that $g_i.x_i$ is the normal form $F_i(x_i)$. Given input $x\in\cS$, the steps are:
\begin{enumerate}
    \item[1.] For each $1\leq i\leq n$ find $g_i\in G_i$ such that $F_i(\pi_i(x))=g_i.\pi_i(x)$.
    \item[2.] Set $\omega = (g_1^{-1},\dots,g_n^{-1}).x$.
\end{enumerate}
This fact is used to construct \cref{alg:OHOSVD,alg:slocc-odd} later in the text.

\begin{proposition}\label{core2}
    Elements $x,y\in \cS$ have equivalent core elements if and only if $x$ and $y$ are in the same $G_1\times\dots\times G_n$-orbit.
\end{proposition}
\begin{proof}
    It is clear that if $x$ and $y$ have equivalent core elements, then they are in the same orbit. Conversely, suppose $(g_1,\dots, g_n).x=y$ for some $g_i\in G_i$. Set $N_i=F_i(\pi_i(x))=F_i(\pi_i(y))$. A core element $\omega$ for $x$ has the properties $x=(h_1,\dots, h_n).\omega$ for some $h_i\in G_i$ and $\pi_i(\omega)=N_i$ for all $i$. Then $\omega$ is also a core element for $y$, since $y=(g_1h_1,\dots, g_nh_n).\omega$ and $\pi_i(\omega)=N_i$. By the uniqueness of core elements in \cref{lem:core}, we are done.
\end{proof}

We wish to describe an analogous version of the HOSVD where the complex orthogonal group takes the role of the unitary group. To that end, let us recall two analogous matrix factorizations.
\begin{proposition}[Spectral theorem \cite{Gantmacher}]
     Every $d\times d$ Hermitian matrix $A$ has a factorization $A=UD U^*$ where $U\in \U_d$ and $D$ is diagonal with real entries.
\end{proposition}

\begin{proposition}[{Normal Forms for Complex Symmetric Matrices}{\cite[Section~3]{Gantmacher_1990a}}]\label{prop:od}
    Every $d\times d$ complex symmetric matrix $A$ has a factorization $A=UDU^\top $ where $U\in\Or_d$ and $D=S_{k_1}(\lambda_1)\oplus\dots\oplus S_{k_r}(\lambda_r)$, where $S_k(\lambda)$, $\lambda\in\CC$ is the $k\times k$ \defi{symmetrized Jordan block}
    \[S_k(\lambda)=\begin{pmatrix}
        \lambda & 1 & \cdots  & \cdots & 0\\
        1 & \lambda & \ddots &  & \vdots\\
        \vdots & \ddots & \ddots  & \ddots & \vdots\\
        \vdots &  & \ddots & \lambda & 1\\
        0 & \cdots & \cdots & 1 & \lambda\\
    \end{pmatrix}+
    \im\begin{pmatrix}
        0 & \cdots & \cdots  & 1 & 0\\
        \vdots & & \iddots & 0 & -1\\
        \vdots & \iddots & \iddots  & \iddots & \vdots\\
        1 &  0& \iddots & & \vdots\\
        0 & -1 & \cdots & \cdots & 0\\
    \end{pmatrix}.
    \]
The symmetrized Jordan block $S_k(\lambda)$ is similar to the standard $k\times k$ Jordan block with diagonal entries equal to $\lambda$.
\end{proposition}
The spectral theorem gives diagonal normal forms for the action of $\U_d$ by conjugation on the space of Hermitian matrices, whereas \cref{prop:od} gives block diagonal normal forms for the action of $\Or_d$ by conjugation on the space of complex symmetric matrices. However, one has to fix an order on the diagonal entries or Jordan blocks. In the unitary case we can use the standard ordering on the reals. For the orthogonal case, we use the lexicographical (lex) order on $\CC$ together with comparing sizes of the blocks - more details in \cref{sec:HOSVD,sec:complexOrthoTensor}.

It is easy to see that \cref{prop:od} remains true if $\Or_d$ is replaced with $\SO_d$. If a $d\times d$ complex symmetric matrix $A$ has $d$ distinct eigenvalues, then \cref{prop:od} implies that $A=UDU^\top$ where $U\in\SO_d$ and $D=\text{diag}(\lambda_1,\dots.\lambda_d)$ is diagonal. Given such $A$, one can compute $U$ and $D$ with \cref{alg:orthoSVD}, which is very similar to the standard spectral theorem algorithm. Note that, since the matrix $U$ computed in step 1 has full rank, the diagonal matrix $U^\top U$ has nonzero diagonal entries. Then each $u_i^\top u_i$ in step 2 is not zero.

\begin{lemma}\label{lem:eigorth}
    Let $A\in \CC^{d\times d}$ be symmetric. Suppose $A$ has an eigenvector $v$ with eigenvalue $\lambda$ and an eigenvector $w$ with eigenvalue $\mu$. If $\lambda\neq \mu$, then $v^{\top} w=0$.
\end{lemma}
\begin{proof}
    Since $A$ is symmetric,
    $
    \lambda v^\top w=(Av)^\top w=v^\top A^\top w=v^\top Aw =\mu v^\top w
    $
    which implies $(\lambda -\mu) v^\top w=0$.
\end{proof}
\begin{algorithm}\caption{Orthogonal spectral theorem for general matrices}\label{alg:orthoSVD}
    \noindent \textit{Input}: A symmetric matrix $A\in\CC^{d\times d}$ with distinct eigenvalues.

    \noindent \textit{Output}: Matrices $U\in\SO_d$ and $D\in\CC^{d\times d}$, diagonal with weakly decreasing diagonal entries (in lex order), such that $A=UDU^{\top}$.

    \begin{enumerate}
    \item[1.] Diagonalize $A$, i.e. compute matrices $U\in\GL_d$ and $D\in\CC^{d\times d}$, diagonal with weakly decreasing diagonal entries, such that $A=UDU^{-1}$.
    \item[2.] Applying \cref{lem:eigorth} to the columns of $U$, we see that $U^{\top}U$ is diagonal. Update $U$ by replacing each column vector $u_i$ of $U$ with $u_i/(u_i^{\top} u_i)^{1/2}$, $1\leq i\leq d$. This ensures that $U^{\top}U=I$.
    \item[3.] If $\det(U)=-1$, update $U\gets U\cdot\text{diag}(-1,1,1,\dots,1)$ so that $U\in\SO_d$. 
\end{enumerate}
\end{algorithm}

\subsection{HOSVD}\label{sec:HOSVD}
For $1\leq i\leq n$ let $\U_{d_i}\to \GL(V_i)$ be the $d_i$-dimensional representation of the unitary group. Let $\cS_i$ be the space of $d_i\times d_i$ Hermitian matrices considered as a $\U_{d_i}$-module under the conjugation action, and define maps
\begin{gather*}
    \pi_i\colon V_1\otimes\dots\otimes V_n\to \mathcal{S}_i\quad\text{where}\quad    \pi_i(\Phi)=\Phi_{(i)}\Phi_{(i)}^*
\end{gather*}
for $\Phi\in V_1\otimes\dots\otimes V_n$. By \cref{ex:trout}, $\pi_i$ is a reduction map.
     
    By the spectral theorem, every $M\in \mathcal{S}_i$ can be factored as $M=UD U^*$ for some $U\in\U_{d_i}$ and real diagonal $D$. We uniquely specify $D$ by requiring that its diagonal entries be listed in weakly decreasing order. Define a normal form function $F_i\colon \mathcal{S}_i\to \mathcal{S}_i$ by setting $F_i(M)=D$.
    In this case, \cref{lem:core} recovers the HOSVD described in \cite{DeLathauwer_DeMoor_Vandewalle_2000}. The term ``core element'' comes from this source. We state the result below. The second item is equivalent to the ``all-orthogonality'' condition from \cite{DeLathauwer_DeMoor_Vandewalle_2000}.

\begin{theorem}[{HOSVD, \cite{DeLathauwer_DeMoor_Vandewalle_2000}}]\label{thm:HOSVD}
    For $1\leq i\leq n$ let $V_i$ be a $d_i$-dimensional complex vector space. Then for each tensor $\Phi\in V_1\otimes\dots\otimes V_n$ there exists a core tensor $\Omega$ such that
    \begin{itemize}
        \item $\Phi=(U_1\otimes\dots\otimes U_n)\Omega$ for some $U_i\in\U_{d_i}$, and
        \item $D_i=\Omega_{(i)}\Omega_{(i)}^*$ is real diagonal with weakly decreasing diagonal entries for all $1\leq i\leq n$.
    \end{itemize}
The core tensor is unique up to the action of $H_1\times\dots\times H_n$, where $H_i\leq \U_{d_i}$ is the stabilizer subgroup of $D_i$ by the conjugation action. 
\end{theorem}
\begin{proof}
    Apply \cref{lem:core}.
\end{proof}
\subsection{Orthogonal HOSVD}\label{sec:complexOrthoTensor}
For $1\leq i\leq n$ let $\SO_{d_i}\to\GL(V_i)$ be the $d_i$-dimensional representation of $\SO_{d_i}$. Let $\mathcal{S}_i\cong S^2 V_i$ be the space of $d_i\times d_i$ complex symmetric matrices considered as a $\SO_{d_i}$-module by the conjugation action, and define
\[
\pi_i\colon V_1\otimes\dots\otimes V_n\to\mathcal{S}_i,\quad\text{where}\quad\pi_i(\Phi)=\Phi_{(i)}\Phi_{(i)}^\top 
\]
for $\Phi\in V_1\otimes \dots\otimes V_n$. Then $\pi_i$ is a reduction map.

By \cref{prop:od}, every $M\in\mathcal{S}_i$ can be factored as $M=UDU^\top $ where $U\in \SO_{d_i}$ and $D=J_{k_1}(\lambda_1)\oplus\dots\oplus J_{k_r}(\lambda_r)$ is a direct sum of symmetrized Jordan blocks. We can uniquely specify $D$ by ordering the blocks in weakly decreasing order, where order on the set of blocks $J_{k}(\lambda)$ is induced by lexicographical order on the triples $(\text{Re}(\lambda),\text{Im}(\lambda),k)$. Define a normal form function $F_i\colon \mathcal{S}_i\to\mathcal{S}_i$ by setting $F_i(M)=D$.

With this setup, \cref{lem:core} gives us the following tensor factorization. We call it the \defi{orthogonal HOSVD} since it may be viewed as the complex orthogonal version of the HOSVD.

\begin{theorem}[Orthogonal HOSVD]\label{thm:OHOSVD}
    For $1\leq i\leq n$ let $V_i$ be a $d_i$-dimensional complex vector space. Then for each tensor $\Phi\in V_1\otimes\dots\otimes V_n$ there exists a core tensor $\Omega$ such that
    \begin{itemize}
        \item $\Phi=(U_1\otimes\dots\otimes U_n)\Omega$ for some $U_i\in \SO_{d_i}$, and
        \item $D_i=\Omega_{(i)}\Omega_{(i)}^\top $ is a direct sum of symmetrized Jordan blocks in weakly decreasing order for all $1\leq i\leq n$. 
    \end{itemize}
The core tensor is unique up to the action of $H_1\times\dots\times H_n$, where $H_i\leq \SO_{d_i}$ is the stabilizer subgroup of $D_i$ by the conjugation action.
\end{theorem}
\begin{proof}
    Apply \cref{lem:core}.
\end{proof}

\begin{proposition}\label{diag}
    Let $G$ be a subgroup of $\GL_d$ and let $D=\lambda_1 I_{k_1}\oplus\dots\oplus \lambda_r I_{k_r}$ where the $\lambda_i\in\CC$ are pairwise distinct and $k_1+\dots+k_r=d$.  Then
    \[
    H:=\{U\in G\colon  UDU^{-1}=D\}=G\cap (\GL_{k_1}\oplus\dots\oplus\GL_{k_r}).
    \]
    In particular, $H=\U_{k_1}\oplus\dots\oplus\U_{k_r}$ if $G=\U_d$ and $H=\SO_{k_1}\oplus\dots\oplus\SO_{k_r}$ if $G=\SO_d$.
\end{proposition}
\begin{proof}
    The containment $\supset$ is clear. For the other direction, suppose $U\in G$ such that $UDU^{-1}=D$. Then $U$ commutes with $D$, which implies that $U$ preserves the eigenspaces of $D$. Thus $U$ has block diagonal form $A_1\oplus\dots\oplus A_r$ with $A_i\in\GL_{k_i}$.
\end{proof}

\cref{diag} helps us understand what the stabilizers $H_i$ of \cref{thm:HOSVD,thm:OHOSVD} are. In the unitary case (\cref{thm:HOSVD}) the reduced density matrices $D_i$ are diagonal, so we have a complete understanding of the possible stabilizers. In the orthogonal case (\cref{thm:OHOSVD}), the matrices $D_i$ are generically diagonal. We say more about this situation in \cref{sec:jordan}, but will not give a complete description of the possible stabilizers.

Let $\Phi\in V_1\otimes\dots\otimes V_n$ as in the statement of \cref{thm:OHOSVD}. Suppose that for each $1\leq i\leq n$ the matrix $\Phi_{(i)}\Phi_{(i)}^\top$ has distinct eigenvalues. Then the orthogonal HOSVD core tensor $\Omega$ has the property that $\Omega_{(i)}\Omega_{(i)}^\top$ is diagonal with distinct diagonal entries. In this case, \cref{alg:OHOSVD} for computing $\Omega$ mirrors the classical HOSVD algorithm \cite{DeLathauwer_DeMoor_Vandewalle_2000}; both follow from the proof of \cref{lem:core}. We leave the situation with repeated eigenvalues open. 

\begin{algorithm}
\caption{Orthogonal HOSVD for general tensors}
\label{alg:OHOSVD}
    \:
    \newline
    \noindent \textit{Input}: A tensor $\Phi\in \CC^{d_1}\otimes\dots\otimes \CC^{d_n}$ such that for each $1\leq i\leq n$, $\Phi_{(i)}\Phi_{(i)}^\top$ has distinct eigenvalues.

\noindent \textit{Output}: Core tensor $\Omega$ in the sense of \cref{thm:OHOSVD}.
\begin{enumerate}
    \item[1.] For $1\leq i\leq n$ use \cref{alg:orthoSVD} to factorize $\Phi_{(i)}\Phi_{(i)}^\top=U_iD_i U_i^\top$, where $U_i\in\SO_{d_i}$ and $D_i$ is diagonal with decreasing diagonal entries.
    \item[2.] Set $\Omega \gets (U_1^\top\otimes\dots\otimes U^\top_n)\Phi$. 
\end{enumerate}
\end{algorithm}

\subsection{Stabilizers of Jordan blocks}\label{sec:jordan}

In this section, let $J_k(\lambda)$ denote the Jordan block with diagonal entries equal to $\lambda$. Recall also our notation for the symmetrized Jordan block $S_k(\lambda)$ introduced in \cref{prop:od}.

Let $G\leq \GL_n$. The stabilizer of a matrix $M\in \CC^{n\times n}$ in $G$ is
\[
\Stab_{G}(M) = \{A\in G \: \mid \: AMA^{-1} = M \}.
\]
\cref{diag} described $\Stab_{\SO_n}(M)$ when the Jordan blocks of $M$ are all $1\times 1$. In general, any $A\in \Stab_{G}(M)$ must preserve the flags of generalized eigenspaces of $M$.
We can be more specific in the case that the Jordan form of $M$ consists of a single Jordan block. We note that in general, the description of the stabilizer should be more complicated.

\begin{lemma}\label{lem:up}
Consider the Jordan block $M = J_k(\lambda)$. We have
\[\Stab_{\GL_k}(M) = 
\left\{
\begin{pmatrix}
             a_1 & a_2 & \cdots & a_k \\
             & \ddots & \ddots & \vdots\\
             & & \ddots & a_2 \\
             &  &  & a_1
        \end{pmatrix}\; \middle| \; a_1 \neq 0
\right\}.
\]
\end{lemma}
\begin{proof}
    Suppose $A$ commutes with $J_k(\lambda)$. Then $A$ preserves the generalized eigenspaces of $J_k(\lambda)$, hence $A$ is upper triangular. Since
    \[
        A(\lambda I_k + J_k(0))=AJ_k(\lambda)=J_k(\lambda) A = (\lambda I_k + J_k(0))A,
    \]
    it follows that $AJ_k(0) = J_k(0)A$. We have
    \begin{align}\label{eq:cal}
    \begin{split}
        J_k(0) A e_i &= A_{2,i}e_1 + A_{3,i} e_2 + \dots + A_{i,i} e_{i-1},\quad \text{and} \\
        A J_k(0) e_i &= A_{1,i-1}e_1 + A_{2,i-1} e_2 + \dots + A_{i-1,i-1}e_{i-1}
    \end{split}
    \end{align}
    for all $2\leq i\leq k$. Thus, $A$ has the desired form.

    Conversely, suppose $A$ has the form above. As in the preceding paragraph, it suffices to show that $A$ commutes with $J_k(0)$. This follows from \eqref{eq:cal}.
\end{proof}
In \cref{lem:up}, we showed that $A$ lies in a $k$-dimensional subspace of $\CC^{k\times k}$. The obvious basis for this space can be written as
$\{I_k, J_k(0), J_k(0)^2,\dots, J_k(0)^{k-1}\}.$
By \cref{prop:od}, the symmetrized Jordan block $S_k(\lambda)$ is similar to $J_k(\lambda)$. Since similarity transformations are linear and preserve matrix multiplication, we immediately obtain the following.
\begin{proposition}\label{prop:symjb}
    Consider the symmetrized Jordan block $M=S_k(\lambda)$. We have
    \[
      \Stab_{\GL_k}(M) = \textnormal{span}\{I_k, S_k(0), S_k(0)^2,\dots, S_k(0)^{k-1}\}\cap\GL_n.
    \]
\end{proposition}
\begin{proof}
    Apply \cref{lem:up}.
\end{proof}
\begin{example}
    Let $A\in\GL_4$ and $S=S_4(0)$. We have
    \[
        S = \begin{pmatrix}
            0 & 1 & \im & 0 \\
            1 & \im & 1 & -\im \\
            \im & 1 & -\im & 1 \\
            0 & -\im & 1 & 0
        \end{pmatrix},\quad
        S^2 = \begin{pmatrix}
            0 & 2\im & 2 & 0 \\
            2\im & 0 & 0 & 2 \\
            2 & 0 & 0 & -2\im \\
            0 & 2 & -2\im & 0
        \end{pmatrix},\quad
        S^3 = \begin{pmatrix}
            4\im & 0 & 0 & 4 \\
            0 & 0 & 0 & 0 \\
            0 & 0 & 0 & 0 \\
            4 & 0 & 0 & -4\im
        \end{pmatrix}.
    \]
    In this case, \cref{prop:symjb} says that $A$ stabilizes $S_4(\lambda)$ by the conjugation action if and only if $A$ has the form below, where $a,b,c,d\in\CC$.
    \[
            A = \begin{pmatrix}
            a+ d\im & b+c\im & c+b\im & d \\
            b+c\im & a+b\im & b & c-b\im \\
            c+b\im & b & a-b\im & b-c\im \\
            d & c-b\im & b-c\im & a-d\im
        \end{pmatrix}
    \]
\end{example}

\begin{proposition}
    Let $M=S_k(\lambda)$. Then $\Stab_{\Or_k}(M)=\{\pm I\}$.
\end{proposition}
\begin{proof}
    One direction is obvious. For the other, suppose $U\in \Or_k$ commutes with $M$. By \cref{prop:symjb}, we have
    \[
    U = \alpha_0 I_k + \alpha_1 S_k(0) + \alpha_2 S_k(0)^2 + \dots + \alpha_{k-1} S_k(0)^{k-1}
    \]
    for some constants $\alpha_i\in\CC$. Then since $U$ is symmetric and orthogonal, $U^2=UU^\top=I_k$. It follows that $\alpha_i=0$ when $i>0$, and $\alpha_0=\pm 1$.
\end{proof}

\section{Normal forms for almost all qubit states}\label{sec:normal}
In this section we present several algorithms for finding normal forms with respect to various group actions on $\cH_n$. For the LU case we give a detailed exposition of ideas in \cite{BKraus} that lead to Kraus's \cref{alg:kraus}. Our more general perspective allows us to go further and obtain \cref{alg:slocc-even,alg:slocc-odd} for the SLOCC cases.
\subsection{General points and generic properties}\label{sec:genericProperties}
Let $V$ be a finite-dimensional complex
vector space, endowed with the usual Euclidean topology. A \defi{general point} of $V$ is one contained in an open, full-measure (hence also dense) subset of $V$ . A property that holds for general points is a \defi{generic property}.
\begin{lemma}\label{lem:ae}
    For each $i=1,\dots,r$ let $p_i:W_i\to\CC$ be a nonzero holomorphic function on an open, full-measure subset $W_i\subset V$. Then $D=\{v\in V:p_i(v)\neq 0,\forall i\}$ is an open, full-measure subset of $V$.
\end{lemma}
\begin{proof}
    For each $i$ the set of zeros $Z_i=p_i^{-1}(0)$ of $p_i$ is a null set (see \cite[p.~9]{gunning2022analytic}). Then $W_i\setminus Z_i$ is open and of full measure. Then the intersection $D=\bigcap_{i=1}^r W_i\setminus Z_i$ is open and of full measure.
\end{proof}
According to \cref{lem:ae}, to show that a property $P$ holds for general $\Phi\in V$, it suffices to exhibit a collection $\{p_i:i=1,\dots,n\}$ of complex-valued functions, each holomorphic on an open, dense, full-measure subset of $V$, such that $P$ holds for every $\Phi\in V$ such that $p_i(\Phi)\neq 0$ for all $i$. If the functions $p_i$ are polynomials in $\CC[V]$, then $P$ is generic in the stronger sense of holding on a Zariski open subset of $V$. This is the case when $P$ is the property that a matrix $M\in \CC^{d\times d}$ does not have repeated eigenvalues since $P$ holds if and only if the discriminant of the characteristic polynomial of $M$ does not vanish.

\subsection{LU group}\label{sec:tpt}
The HOSVD of tensors gets us halfway to finding normal forms for almost all tensors under the LU action. Suppose that $\Phi\in\cH_n$ is general: specifically, assume that $\Phi_{(i)}\Phi_{(i)}^*$ has distinct eigenvalues for all $1\leq i\leq n$. Then the HOSVD core tensor $\Omega$ associated to $\Phi$ has the property that $\Omega_{(i)}\Omega_{(i)}^*$ is real diagonal with distinct diagonal entries.
By \cref{diag}, the core tensor $\Omega$ is unique up to the action of $H^{\times n}$, where 
\[
H=\{\text{diag}(e^{\im s},e^{\im t}):s,t\in\RR\}.
\]
So the problem is reduced to finding an easily computable normal form $\Omega'$ in the $H^{\times n}$-orbit of $\Omega$ since such $\Omega'$ also serves as a normal form for the $\U_2^{\times n}$-orbit of $\Phi$. First we present \cref{alg:Tnnormal}, which is a simple approach for achieving this, but which assumes that $\Omega$ is not too sparse. Next we present Kraus's solution \cite{BKraus}.

By pulling out scalars, the action of $H^{\times n}$ is the same as the action of $\
\mathbb{T}^{\times n}$ plus multiplication by a global scalar $e^{\im s}$, where $s\in\RR$ and
\[
\mathbb{T}=\{\text{diag}(1,e^{\im t}):t\in\RR\}.
\]
The action of $\mathbb{T}^{\times n}$ can be understood through the following formula for basis vectors, which extends to $\cH_n$ by linearity:
\begin{equation}\label{hbar}
\begin{pmatrix}
    1 & \\
    & e^{\im t_1}\\
    \end{pmatrix}\otimes
    \dots
    \otimes
    \begin{pmatrix}
    1 & \\
    & e^{\im t_n}\\
    \end{pmatrix}
    \ket{v_1\dots v_n}=\exp({\im (t_1v_1+\dots+t_nv_n)})\ket{v_1\dots v_n}
\end{equation}
for each $(v_1,\dots,v_n)\in\{0,1\}^n$.

\subsubsection{Simple normal form under unitary stabilizers}\label{sec:simple}
    For $1\leq i\leq n$ let $\textbf{v}^i\in\{0,1\}^n$ be the ``basis vector'' $\textbf{v}^i=(\delta_{1i},\dots,\delta_{ni})$ with 1 in the $i$th entry and 0's elsewhere, then let $\mathcal{B}=\{\textbf{v}^1,\dots,\textbf{v}^n\}$. \cref{alg:Tnnormal} provides a simple way to compute the $H^{\times n}$-normal form $\Omega'$ for $\Omega$ assuming that $\Omega_{\textbf{v}}\neq 0$ whenever $\textbf{v}\in\mathcal{B}\cup\{\textbf{0}\}$.
    
\begin{algorithm}\caption{Simple normal form under unitary stabilizers}\label{alg:Tnnormal}
\noindent \textit{Input}: $\Omega\in\cH_n$ such that $\Omega_{\textbf{v}}\neq 0$ whenever $\textbf{v}\in\mathcal{B}\cup\{\textbf{0}\}$.

\noindent \textit{Output}: The unique $\Omega'$ in the $H^{\times n}$-orbit of $\Omega$ such that each entry $\Omega'_{\textbf{v}}$ is real and positive whenever $\textbf{v}\in\mathcal{B}\cup\{0\}$.

\begin{enumerate}
    \item[1.] Update $\Omega\gets e^{\im t}\Omega$, where $t\in\RR$ such that $e^{\im t}\Omega_\textbf{0}$ is real and positive.
    \item[2.] For $1\leq i\leq n$ choose $t_i\in\RR$ so that $e^{\im  t_i}\Omega_{\textbf{v}^i}$ is real and positive.
    \item[3.] Compute $
    \Omega'=
    \begin{pmatrix}
    1 & \\
    & e^{\im t_1}\\
    \end{pmatrix}\otimes
    \dots
    \otimes
    \begin{pmatrix}
    1 & \\
    & e^{\im t_n}\\
    \end{pmatrix}
    \Omega.
    $ 
\end{enumerate}
\end{algorithm}

\subsubsection{Kraus's algorithm}{\cite{BKraus}}\label{sec:kraus}
    \cref{alg:kraus} produces an $H^{\times n}$-normal form $\Omega'$ for any $0\neq \Omega\in\cH_n$. The following notation is used. Let $\supp(\Omega)=\{\textbf{v}:\Omega_\textbf{v}\neq 0\}$ denote the support. We may assume without loss of generality that $\textbf{0}\in\supp(\Omega)$, i.e. that $\Omega_\textbf{0}\neq 0$, since $\Omega$ must have some nonzero entry and we can relabel the basis vectors $\ket{0}$ and $\ket{1}$ in the tensor factors of $\cH_n=(\CC^2)^{\otimes n}$. The function arg returns the argument of a complex number.

\begin{algorithm}\caption{Kraus's algorithm}\label{alg:kraus}
\:\newline
    \noindent \textit{Input}: $\Omega\in\cH_n$ such that $\Omega_{\textbf{0}}\neq 0$.

    \noindent \textit{Output}: The unique $\Omega'$ in the $H^{\times n}$-orbit of $\Omega$ such that each entry $\Omega'_{\textbf{v}}$ is real and positive whenever $\textbf{v}\in\mathcal{B}\cup\{0\}$, where $\mathcal{B}\subset\supp(\Omega)$ is constructed from $\supp(\Omega)$ by the algorithm.

    \begin{enumerate}
        \item[1.] Update $\Omega\gets e^{\im t}\Omega$, where $t\in\RR$ such that $e^{\im t}\Omega_\textbf{0}$ is real and positive.
        \item[2.] Construct $\mathcal{B}=\{\textbf{v}^1,\dots,\textbf{v}^m\}$ as follows. First set $\mathcal{B}\gets\emptyset$. Then, going over elements $\textbf{v}\in\supp(\Omega)\setminus\{\textbf{0}\}$ in increasing lex order append $\textbf{v}$ to $\mathcal{B}$ if $\textbf{v}$ is linearly independent over $\mathbb{R}$ from the vectors already in $\mathcal{B}$. Stop once $\mathcal{B}$ spans the same space as $\supp(\Omega)$.
        \item[3.]\label{Kstep3} Compute any row vector $\textbf{t}=(t_1,\dots, t_n)\in\RR^n$ satisfying the system \[
        \begin{pmatrix} t_1 & \dots & t_n\end{pmatrix}\begin{pmatrix} \textbf{v}^1 & \dots & \textbf{v}^m\end{pmatrix}=-\begin{pmatrix}
        \text{arg}(\Omega_{\textbf{v}^1}) &
        \dots &
        \text{arg}(\Omega_{\textbf{v}^m}) 
        \end{pmatrix}.
        \]
        \item[4.] Compute $\Omega'=
        \begin{pmatrix}
        1 & \\
         & e^{\im t_1}
        \end{pmatrix}
        \otimes
        \dots
        \otimes
        \begin{pmatrix}
            1 & \\
             & e^{\im t_n}
        \end{pmatrix}
        \Omega$. 
    \end{enumerate}
    \end{algorithm}
    From step 3 we have $\textbf{t}\textbf{v}^i=t_1\textbf{v}^i_1 + \dots + t_n\textbf{v}^i_n=-\text{arg}(\Omega_{\textbf{v}^i})$ for all $1\leq i \leq m$. Then applying Eq.~\eqref{hbar},
    \[
    \Omega'=\sum_{\textbf{v}\in \mathcal{B}} \exp({-\im \,\text{arg}(\Omega_\textbf{v})}) \Omega_\textbf{v}\ket{\textbf{v}}+\sum_{\textbf{v}\in \supp(\Omega)\setminus \mathcal{B}}  \Omega'_\textbf{v}\ket{\textbf{v}}.
    \]
    This expression for $\Omega'$ shows that the coefficients corresponding to $\textbf{v}\in\mathcal{B}$ are real and positive as claimed. By construction, $\mathcal{B}$ is a basis for the span of $\supp(\Omega)$. Then every $\textbf{v}\in\supp(\Omega)$ is a unique real linear combination $\sum_{i=1}^m\alpha_i \textbf{v}^i$ of vectors $\textbf{v}^i\in\mathcal{B}$. If $\Omega$ is the tensor obtained after step 3, then
    \begin{align*}
    \Omega'_\textbf{v}&= \exp\left(\im \textbf{t}\textbf{v}\right)\Omega_\textbf{v}
    = \exp\left (\im \sum \alpha_i\textbf{t}\textbf{v}^i\right)\Omega_\textbf{v}= \exp\left(-\im \sum\alpha_i\arg (\Omega_{\textbf{v}^i})\right)\Omega_\textbf{v}
    \end{align*}
    for every $\textbf{v}$ in the support. Therefore any $\textbf{t}$ satisfying the system in step 3 produces the same $\Omega'$. It follows that $\Omega'$ is unique. When $m<n$, the solution $\textbf{t}$ is not unique; this means that $\Omega'$ has a nontrivial stabilizer in $\GL_2^{\times n}$.

\subsection{SLOCC group, even case}\label{sec:sleven}
    Recall that $T(A\otimes B)T^*\in\SO_4$
    when $A,B\in\SL_2$ (see \cref{ex:li}).
     Consider the natural action of $\SL_2(\CC)\times \SL_2(\CC)$ on the vector space $\CC^2\otimes\CC^2$. Applying the $k$-fold Kronecker product $T^{\otimes k}$ we obtain
    \begin{equation}\label{slso}
    \tilde{O}T^{\otimes k}\Phi=T^{\otimes k}\tilde{A}\Phi,\quad \forall \Phi\in (\CC^2\otimes\CC^2)^{\otimes k}
    \end{equation}
    where
    $\tilde{A}=(A_1\otimes B_1)\otimes\dots\otimes (A_k\otimes B_k)\in (\SL_2\otimes\SL_2)^{\otimes k}
    $
    is the Kronecker product of $\SL_2$ matrices and $\tilde{O}=T^{\otimes k}\tilde{A}T^{*\otimes k}\in \SO_4^{\otimes k}$ is the Kronecker product of $\SO_4$ matrices. We immediately get the following results.
    \begin{lemma}\label{orbe}
        Two tensors $\Phi,\Phi'\in (\CC^2\otimes\CC^2)^{\otimes k}$ are in the same $(\SL_2\times\SL_2)^{\times k}$-orbit if and only if $T^{\otimes k}\Phi$ and $T^{\otimes k}\Phi'$ are in the same $\SO_4^{\times k}$-orbit. 
    \end{lemma}
    \begin{proof}
        Follows from Eq.~\eqref{slso}.
    \end{proof}
    
\begin{theorem}\label{thm:sleven} Let $k>1$.
    For each tensor $\Phi\in (\CC^2\otimes\CC^2)^{\otimes k}$ there exists a core tensor $\Omega\in (\CC^2\otimes\CC^2)^{\otimes k}$ such that
    \begin{itemize}
        \item $\Phi=(A_1\otimes B_1\otimes\dots\otimes A_k\otimes B_k)\Omega$ for some $A_i,B_i\in\SL_2$, and
        \item $D_i=(T^{\otimes k}\Omega)_{(i)}(T^{\otimes k}\Omega)_{(i)}^\top $ is a direct sum of symmetrized Jordan blocks in weakly decreasing order for all $1\leq i\leq k$.
    \end{itemize}
The core tensor is unique up to the action of $H_1\times\dots\times H_k$, where $H_i\leq \SL_2\otimes\SL_2$ such that $TH_i T^*\leq\SO_4$ is the stabilizer of $D_i$ by the conjugation action.
\end{theorem}
\begin{proof}[Proof~1]
Given $\Phi\in (\CC^2\otimes\CC^2)^{\otimes k}$ let $\Psi=T^{\otimes k}\Phi$. Applying \cref{thm:OHOSVD} to $\Psi$ in $(\CC^2\otimes\CC^2)^{\otimes k}$ as an $\SO_4^{\times k}$-module, we get a core tensor $\Omega$ with certain properties. These properties are equivalent to the ones above due to \cref{orbe} and the isomorphism \eqref{iso}.
\end{proof}

\begin{proof}[Proof~2]For $1\leq i\leq k$ let $V_i$ be a copy of the $\SL_2\times \SL_2$-module $\CC^2\otimes\CC^2$. Let $\cS_i$ be the space of $4\times 4$ complex symmetric matrices, considered as an $\SL_2\times\SL_2$-module by the same representation as \eqref{eq:symrep} with $\cS_i$ replacing $\cS_{ij}$.
By \cref{ex:li} the map
\begin{gather*}
    \pi_i:V_1\otimes\dots\otimes V_k\to \cS_i\quad\text{where}\quad
    \pi_i(\Phi)=T\Phi_{(i)}J^{\otimes {2(k-1)}}\Phi_{(i)}^\top T^\top
\end{gather*}
is a reduction map.
Observing the alternate expression $\pi_i(\Phi)=(T^{\otimes k}\Phi)_{(i)}(T^{\otimes k}\Phi)_{(i)}^\top$ and applying \cref{lem:core}, we immediately have the existence of a core tensor $\Omega$ with the desired properties.
\end{proof}

\subsubsection{Reducing the group action}\label{sec:reduceGroup} Let $k>1$.
Our goal is to compute SLOCC normal forms for general $\Phi\in \cH_{2k}\cong (\CC^2\otimes\CC^2)^{\otimes k}$.
By \cref{orbe}, this problem is equivalent to computing normal forms for the action of $\SO_4^{\times k}$. That is, we could define the SLOCC normal form $\Omega$ of $\Phi$ by the property that $T^{\otimes k}\Omega$ is the $\SO_4^{\times k}$-normal form of $T^{\otimes k}\Phi$ and vice versa. With this in mind, we now focus on the $\SO_4^{\times k}$-action.

Suppose $\Phi\in(\CC^2\otimes\CC^2)^{\otimes k}$ such that $\Phi_{(i)}\Phi_{(i)}^\top$ has distinct eigenvalues for all $1\leq i\leq k$. Let $\Omega$ be the corresponding core tensor in the sense of orthogonal HOSVD, which will have the property that $\Omega_{(i)}\Omega_{(i)}^\top$ is diagonal with distinct diagonal entries. By \cref{diag}, the core tensor $\Omega$ is unique up to the action of $H^{\times k}$, where
\[
H = \{D\in\SO_4 : \text{$D$ is diagonal}\}.
\]
Now the problem is reduced to finding an easily computable $H^{\times k}$-normal form $\Omega'$ for $\Omega$ since such $\Omega'$ also serves as an $\SO_4^{\times k}$-normal form for $\Phi$. As in the LU case, we present two algorithms for achieving this. The first (\cref{alg:sim}) is simpler but assumes $\Omega$ is not too sparse. The second (\cref{alg:enorm}) is modeled after Kraus's \cref{alg:kraus}.

Another way to understand $H$ is as the group of diagonal matrices $D\in \CC^{4\times 4}$ with diagonal entries from $\{\pm 1\}$, since $DD^\top=I$, and an even number of appearances of $-1$, since $\det(D)=1$.
These matrices can be written as
\[
H = \big\{
\pm I_4,\quad \pm Z\otimes I_2,\quad  \pm I_2\otimes Z,\quad \pm Z\otimes Z
\big\}
\]
where $Z=\begin{pmatrix}
    1 & 0 \\ 0 & -1
\end{pmatrix}$ is the Pauli matrix.
Set $n=2k$. From this expression of $H$ we see that the action of $H^{\times k}$ on $(\CC^2\otimes\CC^2)^{\otimes k}$ is equivalent to the action of $\mathcal{T}^{\times n}$ on $\mathcal{H}_n=(\CC^2)^{\otimes n}$, where $\mathcal{T}=\{\pm I_2,\: \pm Z\}$.
The action of $\mathcal{T}^{\times n}$ can be understood by the following formula for basis vectors, extending to $\cH_n$ by linearity:
\begin{equation}\label{tbar}
\begin{pmatrix}
    1 & \\
    & (-1)^{t_1}\\
    \end{pmatrix}\otimes
    \dots
    \otimes
    \begin{pmatrix}
    1 & \\
    & (-1)^{t_n}\\
    \end{pmatrix}
    \ket{v_1\dots v_n}=(-1)^{t_1v_1+\dots+t_nv_n}\ket{v_1\dots v_n}
\end{equation}
for $t_i\in\{0,1\}$ and $(v_1,\dots,v_n)\in\{0,1\}^n$. While the case $n=2k$ is especially important for us, in \cref{sec:simp,sec:gen} $n$ may be even or odd.

\subsubsection{Simple normal form under orthogonal stabilizers}\label{sec:simp}

For $1\leq i\leq n$ let $\textbf{v}^i\in\{0,1\}^n$ be the ``basis vector'' $\textbf{v}^i=(\delta_{1i},\dots,\delta_{ni})$ with 1 in the $i$th entry and 0's elsewhere, then let $\mathcal{B}=\{\textbf{v}^1,\dots,\textbf{v}^n\}$. \cref{alg:sim} computes the $\mathcal{T}^{\times n}$-normal form $\Omega'$ for $\Omega$ with the assumption that $\text{Re}(\Omega_{\textbf{v}})\neq 0$ whenever $\textbf{v}\in\mathcal{B}\cup\{\textbf{0}\}$.

\begin{algorithm}\caption{Simple normal form under orthogonal stabilizers}\label{alg:sim}
    \noindent \textit{Input}: $\Omega\in\cH_n$ such that $\text{Re}(\Omega_{\textbf{v}})\neq 0$ whenever $\textbf{v}\in\mathcal{B}\cup\{\textbf{0}\}$.

    \noindent \textit{Output}: The unique $\Omega'$ in the $\mathcal{T}^{\times n}$-orbit of $\Omega$ such that $\text{Re}(\Omega'_{\textbf{v}})>0$ whenever $\textbf{v}\in\mathcal{B}\cup\{0\}$.
    \begin{enumerate}
        \item[1.] Update $\Omega\gets (-1)^{t}\Omega$, where $t\in\{0,1\}$ such that $(-1)^t \text{Re}(\Omega_\textbf{0})$ is positive.
        \item[2.] For $1\leq i\leq n$ choose $t_i\in\{0,1\}$ such that $(-1)^{t_i}\text{Re}(\Omega_{\textbf{v}^i})$ is positive.
        \item[3.]  Compute $
        \Omega'\gets
        \begin{pmatrix}
        1 & \\
        & (-1)^{t_1}\\
        \end{pmatrix}\otimes
        \dots
        \otimes
        \begin{pmatrix}
        1 & \\
        & (-1)^{t_n}\\
        \end{pmatrix}
        \Omega.$ 
    \end{enumerate}
\end{algorithm}
    
\subsubsection{General normal form under orthogonal stabilizers}\label{sec:gen} \cref{alg:enorm} produces a $\mathcal{T}^{\times n}$-normal form $\Omega'$ for any $0\neq \Omega\in\cH_n$.
Let $\supp(\Omega)=\{\textbf{v}:\Omega_\textbf{v}\neq 0\}$. We may assume without loss of generality that $\textbf{0}\in\supp(\Omega)$, i.e. that $\Omega_\textbf{0}\neq 0$ since $\Omega$ must have some nonzero entry and we can relabel the basis vectors $\ket{0}$ and $\ket{1}$ in the tensor factors of $\cH_n=(\CC^2)^{\otimes n}$. Define a function $s:\CC\setminus\{0\}\to\{0,1\}$ by
\[
s(z) = \begin{cases}
    0 & \text{if $\text{Re}(z)>0$ or if $\text{Re}(z)=0$ and $\text{Im}(z)>0$}\\
    1 & \text{if $\text{Re}(z)<0$ or if $\text{Re}(z)=0$ and $\text{Im}(z)<0$}
\end{cases}
\]
for $z\in \CC\setminus\{0\}$. We think of $s$ as detecting the ``sign'' of $z$.

\begin{algorithm}\caption{General normal form under orthogonal stabilizers}\label{alg:enorm}
\noindent \textit{Input}: A tensor $\Omega\in\cH_n$.

\noindent \textit{Output}: The unique $\Omega'$ in the $\mathcal{T}^{\times n}$-orbit of $\Omega$ such that $s(\Omega'_{\textbf{v}})=0$ whenever $\textbf{v}\in\mathcal{B}\cup\{0\}$, where $\mathcal{B}\subset\supp(\Omega)$ is constructed from $\supp(\Omega)$ by the algorithm.

\begin{enumerate}
    \item[1.] Update $\Omega\gets (-1)^t \Omega$, where $t\in\{0,1\}$ such that $s((-1)^t \Omega_\textbf{0})=0$.
    \item[2.] Construct $\mathcal{B}=\{\textbf{v}^1,\dots,\textbf{v}^m\}$ as follows. First set $\mathcal{B}\gets \emptyset$. Then, going over elements $\textbf{v}\in\supp(\Omega)\setminus\{\textbf{0}\}$ in increasing lex order append $\textbf{v}$ to $\mathcal{B}$ if $\textbf{v}$ is linearly independent over $\mathbb{F}_2$ from the vectors already in $\mathcal{B}$. Stop once $\mathcal{B}$ spans the same space as $\supp(\Omega)$.
    \item[3.] Compute any row vector $(t_1,\dots, t_n)$ over $\mathbb{F}_2$ satisfying the system \[
    \begin{pmatrix} t_1 & \dots & t_n\end{pmatrix}\begin{pmatrix} \textbf{v}^1 & \dots & \textbf{v}^m\end{pmatrix}=\begin{pmatrix}
    s (\Omega_{\textbf{v}^1}) &
    \dots &
    s(\Omega_{\textbf{v}^m}) 
    \end{pmatrix}.
    \]
    \item[4.] Compute $\Omega'\gets
    \begin{pmatrix}
    1 & \\
     & (-1)^{t_1}
    \end{pmatrix}
    \otimes
    \dots
    \otimes
    \begin{pmatrix}
        1 & \\
         & (-1)^{t_n}
    \end{pmatrix}
    \Omega$.
\end{enumerate}
\end{algorithm}

From step 3 of \cref{alg:enorm} we have $t_1\textbf{v}^i_1 + \dots + t_n\textbf{v}^i_n\equiv s(\Omega_{\textbf{v}^i})$ (mod $2$) for all $1\leq i \leq m$. Then applying Eq.~\eqref{tbar},
\[
\Omega'=\sum_{\textbf{v}\in \mathcal{B}} (-1)^{s(\Omega_\textbf{v})} \Omega_\textbf{v}\ket{\textbf{v}}+\sum_{\textbf{v}\in \supp(\Omega)\setminus \mathcal{B}}  \Omega'_\textbf{v}\ket{\textbf{v}}
\]
which shows that $s(\Omega'_{\textbf{v}})=0$ whenever $\textbf{v}\in\mathcal{B}\cup\{0\}$ as claimed. By construction, $\mathcal{B}$ is a basis for the span of $\supp(\Omega)$. Then every $\textbf{v}\in\supp(\Omega)$ is a unique linear combination of vectors in $\mathcal{B}$. From this fact, together with Eq.~\eqref{tbar}, it follows that $\Omega'$ is unique (by an argument similar the one given in \cref{sec:kraus}).

\subsubsection{The even SLOCC normal form algorithm}\label{sec:evennf} 
For clarity, in \cref{alg:slocc-even} we list the full steps of the procedure to compute SLOCC normal forms for general $\Phi\in\cH_n$, with $n=2k\geq 4$ even.

\begin{algorithm}\caption{SLOCC normal form for general qubits, even case}\label{alg:slocc-even}
    \noindent \textit{Input}: A tensor $\Phi\in\cH_{2k}\cong (\CC^2\otimes\CC^2)^{\otimes k}$ with $k>1$ such that for each $1\leq i\leq k$ the matrix $(T^{\otimes k}\Phi)_{(i)}(T^{\otimes k}\Phi)_{(i)}^\top$ has distinct eigenvalues.

    \noindent \textit{Output}: Normal form $\Omega$ in the SLOCC orbit of $\Phi$.

    \begin{enumerate}
        \item[1.] Set $\Phi'\gets T^{\otimes k}\Phi.$
        \item[2.] Use \cref{alg:OHOSVD} to compute a core tensor $\Omega'$ for $\Phi'$ in the sense of \cref{thm:OHOSVD}.
        \item[3.] Use \cref{alg:enorm} to compute the normal form $\Omega''$ in the $\mathcal{T}^{\times 2k}$-orbit of $\Omega'$.
        \item[4.] Set $\Omega \gets {T^*}^{\otimes k}\Omega''$. 
    \end{enumerate}
\end{algorithm}

\subsubsection{The 4-qubit case}\label{sec:4qubit}
Let $V=\CC^2\otimes\CC^2$. Recall that, by \cref{orbe}, the problem of classifying $(\SL_2\times\SL_2)^{\times 2}$-orbits in $V\otimes V$ is equivalent to that of classifying $\SO_4\times\SO_4$-orbits in $V\otimes V$.
More generally, let us consider $\SO_{d_1}\times\SO_{d_2}$-orbits in $V_1\otimes V_2$, where $d_1=\dim V_1$ and $d_2=\dim V_2$. Assume without loss of generality that $d_1\leq d_2$.

\begin{lemma}\label{lem:ABBA}
    Let $A\in \CC^{d_1\times d_2}$ and $B\in \CC^{d_2\times d_1}$. Then $AB$ and $BA$ have the same nonzero eigenvalues, counting multiplicity.
\end{lemma}
\begin{proof}
    See Theorem 2.8 in {\cite{zhang2011matrix}}.
\end{proof}

Suppose $M\in\CC^{d_1\times d_2}$ is general in the sense that $M M^\top$ has distinct eigenvalues, none of them 0. By \cref{lem:ABBA}, $M^\top M$ has $d_1$ nonzero eigenvalues, so the rank of $M^\top M$ is $d_1$ and $M^\top M$ is semisimple (otherwise the rank would exceed $d_1$). A tensor $\Phi\in V_1\otimes V_2$ corresponds to the $d_1\times d_2$ matrix $M=\Phi_{(1)}$.
By this correspondence, \cref{thm:OHOSVD} implies that there exists a core matrix $\Omega\in \CC^{d_1\times d_2}$ such that
\begin{itemize}
    \item $M=U_1 \Omega U_2$, where $U_1\in\SO_{d_1}$ and $U_2\in\SO_{d_2}$,
    \item $\Omega\Omega^\top = D_1$ and $\Omega^\top\Omega = D_2$, where $D_1$ and $D_2$ are diagonal.
\end{itemize}
From the second property we have $D_1\Omega=\Omega\Omega^\top\Omega=\Omega D_2$. This means that multiplying the rows of $\Omega$ by the eigenvalues of $D_1$ is the same as multiplying the columns of $\Omega$ by the eigenvalues of $D_2$. By \cref{lem:ABBA}, $D_1$ and $D_2$ have the same nonzero eigenvalues. By choice of $M$, there are $d_1$ nonzero eigenvalues and they are distinct. From this we conclude that $\Omega$ has at most one nonzero entry in each row and in each column. Acting by signed permutation matrices $P\in \SO_{d_1}$ and $Q\in\SO_{d_2}$ respectively on the left and right, we can permute rows and columns to obtain a matrix of the form
\[
P\Omega Q=
\begin{pmatrix}
\lambda_1 & \dots & 0 & 0 & \dots & 0\\
\vdots & \ddots &  \vdots & \vdots & & \vdots\\
0 & \dots & \lambda_{d_1} & 0 & \dots & 0\\
\end{pmatrix}
\]
where $\lambda_1,\dots,\lambda_{d_1}\in\CC$. Thus, the orbit of a general $M\in \CC^{d_1\times d_2}$ intersects the $d_1$-dimensional space of such matrices.

This result is a special case of \cite[Theorem~2.10]{ChtDjo:NormalFormsTensRanksPureStatesPureQubits}. The theorem of Chterental and Djokovic classifies the various non-general orbits as well, whereas our work generalizes to an arbitrary number of tensor products of $\CC^2$.

\subsubsection{Failure in a special case}\label{sec:specialFailure}
Let us work out an example to illustrate what can happen when $\Phi$ fails the genericity condition. Consider the matrix multiplication tensor
\[
\Phi=\sum_{i,j,k=0}^1 \ket{ik}\otimes \ket{ij}\otimes\ket{jk}\in (\CC^2\otimes\CC^2)^{\otimes 3},
\]
named so because $\Phi$ corresponds to the bilinear form $\CC^{2\times 2}\times \CC^{2\times 2}\to \CC^{2\times 2}$ mapping $(A,B)\mapsto AB$. A computation shows that $\pi_i(\Phi)=2I_4$ for $i=1,2,3$, where $\pi_i(\Phi)=(T^{\otimes 3}\Phi)_{(i)}(T^{\otimes 3}\Phi)_{(i)}^\top$. Thus $\Phi$ is a core tensor for itself in the sense of \cref{thm:sleven}. We do not consider $\Phi$ to be general since every $\pi_i(\Phi)$ has only one eigenvalue, not counting multiplicity. The stabilizer subgroup of $2I_4$ in $\SO_4$ is the entire group $\SO_4$. Then \cref{thm:sleven} tells us that the core tensor is unique up to the action of $H^{\times 3}$, where $H=T^* \SO_4 T=\SL_2\otimes\SL_2.$ In other words, every tensor in the SLOCC orbit of $\Phi$ is a core tensor. This represents the worst-case scenario: \cref{thm:sleven} does not reduce the problem of detecting tensors in the SLOCC orbit of the matrix multiplication tensor to the action of a smaller group, let alone a finite group.

\subsection{SLOCC group, odd case}\label{sec:odd}
Let $n$ be odd and for $1\leq i\leq n$ let $V_i$ be a copy of the $\SL_2$-module $\CC^2$. Let $\cS_i\cong S^2 V_i$ be the space of $2\times 2$ complex symmetric matrices considered as an $\SL_2$-module by the action $A.M=AMA^\top$ for $A\in\SL_2$ and $M\in\cS_i$. For each $i$ define the map $\pi_i$ as follows:
\begin{equation}\label{eq:oddtrout}
\begin{gathered}\xymatrix@R-1.5pc{
\pi_i:V_1\otimes\dots\otimes V_n \ar[r]&\cS_i\cong S^2 V_i
\\
\Phi \ar@{|->}[r]& \Phi_{(i)}J^{\otimes (n-1)}\Phi_{(i)}^\top.
}\end{gathered}
\end{equation}
By \cref{ex:slt}, the map $\pi_i$ is a reduction map.

To find normal forms in $\mathcal{S}_i$, we use the fact that any complex symmetric matrix $M\in\mathcal{S}_i$ admits a factorization $M=ADA^\top $ where $A\in \GL_2$ and $D=I_2$, $\text{diag}(1,0)$ or $0$. Indeed, every nondegenerate complex quadratic form in the variables $x_1,\dots,x_n$ is equivalent to the form $x_1^2+\dots+x_n^2$; see \cite[pp. 97-98]{FultonHarris}. If we require that $A$ is in the smaller group $\SL_2$, we have a factorization $M=ADA^\top$ where $D$ equals $z I_2$, $\text{diag}(z,0)$ or 0 for some $z\in\mathbb{C}$ unique up to multiplication by $-1$. The parameter $z$ is not unique because its argument can be flipped by the action of the matrix
\begin{equation}\label{eq:k}
K=\begin{pmatrix} 0 & \im \\ \im & 0 \end{pmatrix}\in\SL_2.
\end{equation}
Thus, we define a normal form function $F\colon \mathcal{S}_i\to\mathcal{S}_i$ by setting $F(M)=zI_2$, where $z$ is the greater of the two choices in lexicographical order on $\CC$.

Alternatively, one can write $\pi_i(\Phi)$ as $\Phi_{(i)}{T^*}^{\otimes (n-1)/2}(\Phi_{(i)}{T^*}^{\otimes (n-1)/2})^\top$ and prove that $\pi_i$ is a reduction map using the fact that \eqref{iso} is an isomorphism. With this setup, we obtain the following lemma.

\begin{lemma}\label{lem:preodd}
    Let $n$ be odd. For each tensor $\Phi\in\cH_n$ there exists a core tensor $\Psi\in\cH_n$ such that
    \begin{itemize}
        \item $\Phi=(A_1\otimes\dots\otimes A_n)\Psi$ for some $A_i\in\SL_2$, and
        \item $D_i=\Psi_{(i)}J^{\otimes (n-1)}\Psi_{(i)}^\top $ equals $z_i I_2$, $\text{diag}(z_i,0)$ or $0$ where $z_i\in\CC$ and $z_i>-z_i$ in lex order for all $1\leq i\leq n$.
    \end{itemize}
The core tensor $\Psi$ is unique up to the action of $H_1\times\dots\times H_{n}$, where $H_i\leq \SL_2$ is the stabilizer of $D_i$ with respect to the action $A.D_i=AD_i A^\top$ for $A\in \SL_2$.
\end{lemma}
\begin{proof}
    Apply \cref{lem:core}.
\end{proof}
\begin{theorem}\label{thm:odd}
    Let $n\geq 5$ be odd. For general $\Phi\in\cH_n$ there exists a core tensor $\Omega\in \cH_n$, unique up to sign, such that
\begin{itemize}
    \item $\Phi=(A_1\otimes\dots\otimes A_n)\Omega$ for some $A_i\in\SL_2$,
    \item For all $1\leq i\leq n$, $\Omega_{(i)}J^{\otimes (n-1)}\Omega_{(i)}^\top=z_i I_2$ with $z_i>-z_i$ in lex order and $\Omega_{(i)}\Omega_{(i)}^\top =\text{diag}(\lambda_1,\lambda_2)$ with $\lambda_1>\lambda_2$ in lex order.
\end{itemize}
\end{theorem}
\begin{proof}
     Let $\pi_i(\Phi)=\Phi_{(i)}J^{\otimes (n-1)}\Phi_{(i)}^\top$. Suppose $\Phi\in\cH_n$ such that the matrix $\pi_i(\Phi)$ is invertible for all $1\leq i\leq n$. Then, by \cref{lem:preodd}, there exists a core tensor $\Psi$ in the SLOCC orbit of $\Phi$ such that $\pi_i(\Psi)=z_i I_2$ is a scalar multiple of the identity. The stabilizer of $z_i I_2$ in $\SL_2$ is the orthogonal group $\SO_2$. Therefore the core tensor $\Psi$ is unique up to the action of $\SO_2^{\times n}$.

     Suppose further that $\Psi_{(i)}\Psi_{(i)}^\top$ has distinct eigenvalues for each $i$. Then, by \cref{thm:OHOSVD}, there exists a core tensor $\Omega$ in the $\SO_2^{\times n}$-orbit of $\Psi$ (hence in the SLOCC orbit of $\Phi$) such that $\Omega_{(i)}\Omega_{(i)}^\top$ is diagonal with decreasing diagonal entries for all $i$. By \cref{diag}, $\Omega$ is unique up to the action of $H^{\times n}$, where $H=\{\pm I_2\}$. Thus $\Omega$ is unique up to sign.
     
     It remains to prove that the existence of $\Omega$ is generic. That is, $\Psi_{(i)}\Psi_{(i)}^\top$ having distinct eigenvalues must be a generic property of $\Phi\in\cH_n$. We state this fact in the following lemma and postpone the proof to \cref{sec:genlem}.
\end{proof}
\begin{lemma}[Genericity lemma]\label{lem:gen}
    Let $n\geq 5$ be odd, $\Phi\in\cH_n$ be general, and $\Psi$ be the corresponding core tensor in the sense of \cref{lem:preodd}. Then each $\Psi_{(i)}\Psi_{(i)}^\top$ has distinct eigenvalues.
\end{lemma}
\Cref{thm:odd} immediately gives us a way to find normal forms for general $\Phi\in\cH_{n}$ when $n\geq 5$. Given $\Phi$, there exists a core tensor $\Omega$ that is unique up to sign. So pick $\Omega$ or $-\Omega$ to be the normal form; one way to make this selection is described in step 5 of \cref{alg:slocc-odd}. The case where $n=3$ is discussed in \cref{sec:3f}.
\subsubsection{The odd SLOCC normal form algorithm} \label{sec:oddNF}
We now describe an algorithm for computing SLOCC normal forms for general $\Phi\in\cH_n$, with $n\geq 5$ odd. First, we need a way to compute the normal form function $F\colon \mathcal{S}_i\to\mathcal{S}_i$ on the space of $2\times 2$ complex symmetric matrices. If $M\in\mathcal{S}_i$ is invertible, we can do this by $M\mapsto LML^\top=\sqrt{\delta}I_2$, where $\delta=\det(M)$ and $L\in\SL_2$ is given by
\[
L=\begin{dcases}
    \begin{pmatrix}
        M_{11}^{-1/2}\delta^{1/4} & 0 \\
        0 & M_{11}^{1/2}\delta^{-1/4}
    \end{pmatrix}\begin{pmatrix}
        1 & 0 \\
        -M_{12}M_{11}^{-1} & 1
    \end{pmatrix}  & \text{if $M_{11}\neq 0$,} \\
    \begin{pmatrix}
        M_{22}^{1/2}\delta^{-1/4} & 0 \\
        0 & M_{22}^{-1/2}\delta^{1/4}
    \end{pmatrix}\begin{pmatrix}
        1 & -M_{12}M_{22}^{-1} \\
        0 & 1
    \end{pmatrix} & \text{if $M_{22}\neq 0$}, \\
    \begin{pmatrix}
        e^{\im \pi/4} & 0\\
        0 & e^{-\im \pi/4}
    \end{pmatrix}\frac{\im }{\sqrt{2}}\begin{pmatrix}
        -1 & 1 \\
        1 & 1
    \end{pmatrix}& \text{if $M_{11}=M_{22}=0$.} \\
\end{dcases}
\]
If $LML^\top$ is not in normal form (i.e. if $\sqrt{\delta}<-\sqrt{\delta}$ in lex order), then apply $M\mapsto KLML^\top K^\top$, where $K$ is defined in \eqref{eq:k}.

In \cref{alg:slocc-odd}, $\pi_i$ is the map defined in \eqref{eq:oddtrout}.
Steps 1-3 compute the core tensor $\Psi$ in the sense of \cref{lem:preodd}. Step 4 computes the core tensor $\Omega$ in the sense of \cref{thm:odd}. Then step 5 picks out the normal form from the choices $\pm\Omega$.

\begin{algorithm}\caption{SLOCC normal form for general qubits, odd case}\label{alg:slocc-odd}
    \noindent\textit{Input}: A tensor $\Phi\in\cH_{n}$ general in the sense of \cref{thm:odd}, $n\geq 5$ odd.
    
    \noindent\textit{Output}: Normal form $\Omega$ in the SLOCC orbit of $\Phi$.

    \begin{enumerate}
        \item[1.] For $1\leq i\leq n$ use the formula above to compute $L_i\in \SL_2$ such that $L_i\pi_i(\Phi) L_i^\top=\sqrt{\delta_i} I_2$, where $\delta_i=\det(\pi_i(\Phi))$.
        \item[2.] Set $\Psi\gets (L_1\otimes\dots\otimes L_{n})\Phi$ so that $\pi_i(\Psi)=\sqrt{\delta_i} I_2$ for all $i$.
        \item[3.] Update $\Psi\gets (A_1\otimes\dots\otimes A_{n})\Psi$, where $A_i$ equals $K$ [see \eqref{eq:k}] if $\pi_i(\Psi)=\sqrt{\delta_i} I_2$ is not in normal form, i.e. if $\sqrt{\delta_i} < -\sqrt{\delta_i}$ in lex order, otherwise $A_i=I_2$.
        \item[4.] Use \cref{alg:OHOSVD} to compute a core tensor $\Omega$ for $\Psi$ in the sense of \cref{thm:OHOSVD}.
        \item[5.] If the first nonzero entry $a\in\CC$ of $\Omega$ is less than $-a$ in lex order, update $\Omega\gets -\Omega$.  
    \end{enumerate}
\end{algorithm}

\subsubsection{Proof of the genericity lemma}\label{sec:genlem}
We now turn to the proof of \cref{lem:gen}. Let $n>1$ be odd. Suppose $\Phi\in \cH_n$ such that $M_i=\Phi_{(i)}J^{\otimes (n-1)}\Phi_{(i)}^\top$ is invertible and $(M_i)_{11}\neq 0$ for all $1\leq i\leq n$. Feed $\Phi$ into \cref{alg:slocc-odd}. The entries of each matrix $L_i$ computed in step 1 are functions of $M_i$. Moreover, the entries of $M_i$ are quadratic forms in $\CC[\cH_n]$. Therefore the tensor $\Psi = (L_1\otimes\dots\otimes L_{n})\Phi$ computed in step 2 is a function of $\Phi$. In step 4, \cref{alg:OHOSVD} is applied; for this to be possible, the discriminant of the characteristic polynomial
\[
p_i(\Psi)=\text{disc}(\det(\lambda I_2-\Psi_{(i)}\Psi_{(i)}^\top))=\text{tr}(\Psi_{(i)}\Psi_{(i)}^\top)^2 - 4\det(\Psi_{(i)}\Psi_{(i)}^\top)
\]
must not vanish (this is the distinct-eigenvalues condition in \cref{lem:gen}). This is not affected by step 3. By the discussion above, $p_i:W_i\to \CC$ is a function on an open, dense, full-measure subset $W_i\subset\cH_n$. Moreover, $p_i$ is holomorphic on an open, full-measure subset $W'_i\subset W_i$. If $p_i$ is not identically 0 on $W'_i$, then (by \cref{lem:ae}) \cref{alg:slocc-odd} works for general $\Phi\in\cH_n$ as claimed. To show that $p_i\neq 0$, it suffices to exhibit $\Phi\in\cH_n$ such that, for all $1\leq i\leq n$,
\begin{enumerate}
    \item[1.] $M_{i}=\Phi_{(i)}J^{\otimes (n-1)}\Phi_{(i)}^\top$ is a multiple of the identity, and
    \item[2.] $\Phi_{(i)}\Phi_{(i)}^\top$ has distinct eigenvalues.
\end{enumerate}
The first condition implies that $M_i$ is invertible, $(M_i)_{11}\neq 0$, and $\Psi=(L_1\otimes\dots\otimes L_n)\Phi=(I\otimes\dots\otimes I)\Phi=\Phi$. Then the second condition implies $p_i(\Psi)\neq 0$.

In the calculations to follow, let $n=2k+1$. Given $\textbf{v}=(\textbf{v}_1,\dots,\textbf{v}_n)\in\{0,1\}^n$ let $|\textbf{v}|$ be the number of nonzero entries in $\textbf{v}$. Define $\bar{\textbf{v}}\in \{0,1\}^n$ by
\[
\bar{\textbf{v}}_i=\begin{cases}
    0 & \text{if $\textbf{v}_i = 1$,}\\
    1 & \text{if $\textbf{v}_i = 0$.}
\end{cases}
\]
The \defi{Hamming distance} between $\textbf{v}_1,\textbf{v}_2\in\{0,1\}^n$ is the number of occurrences where $(\textbf{v}_1)_i\neq (\textbf{v}_2)_i$ for $1\leq i\leq n$.
We claim that the following tensor satisfies properties 1 and 2 above:
\[
\Phi = (1-2^{2k-1})\ket{\textbf{0}}+\sum_{\textbf{v}\in\mathcal{E}\setminus\{\textbf{0}\}}\ket{\textbf{v}}
\]
where $\mathcal{E}$ is the set of vectors $\textbf{v}\in\{0,1\}^{n}$ such that $|\textbf{v}|$ is even. Splitting terms of $\Phi$ into two groups depending on whether the first tensor factor is $\ket{0}$ or $\ket{1}$, we write
\[
\Phi = \sum_{\textbf{w}\in\{0,1\}^{2k}} a_\textbf{w} \ket{0}\otimes\ket{\textbf{w}} + \sum_{\textbf{w}\in\{0,1\}^{2k}} b_\textbf{w} \ket{1}\otimes\ket{\textbf{w}}
\]
where each $a_{\textbf{w}},b_{\textbf{w}}$ is equal to $0$, $1$, or $1-2^{2k-1}$. Then the mode-1 flattening reads
\[
\Phi_{(1)}=\begin{pmatrix}
\textbf{a}^\top \\
\textbf{b}^\top
\end{pmatrix}=
\begin{pmatrix}
    a_{\textbf{w}_1} & a_{\textbf{w}_2} & \dots & a_{\textbf{w}_N} \\
    b_{\textbf{w}_1} & b_{\textbf{w}_2} & \dots & b_{\textbf{w}_N}
\end{pmatrix}
\]
where $N=2^{2k}$. The images of the reduction maps are
\[
\Phi_{(1)}^\top J^{\otimes 2k}\Phi_{(1)}=\begin{pmatrix}
    \textbf{a}^\top J^{\otimes 2k}\textbf{a} & \textbf{a}^\top J^{\otimes 2k}\textbf{b} \\
    \textbf{a}^\top J^{\otimes 2k}\textbf{b} & \textbf{b}^\top J^{\otimes 2k}\textbf{b} \\
\end{pmatrix}\quad\text{and}\quad
\Phi_{(1)}^\top \Phi_{(1)}=\begin{pmatrix}
    \textbf{a}^\top\textbf{a} & \textbf{a}^\top\textbf{b} \\
    \textbf{a}^\top\textbf{b} & \textbf{b}^\top\textbf{b} \\
\end{pmatrix}.
\]
The operator $J$ maps $\ket{0}\mapsto -\ket{1}$ and $\ket{1}\mapsto\ket{0}$. It follows that $J^{\otimes 2k}\ket{\textbf{w}}=(-1)^{|\bar{\textbf{w}}|} \ket{\bar{\textbf{w}}}$ and we find
\begin{align*}
\textbf{a}^\top J^{\otimes 2k}\textbf{b}&=\left(\sum_{i=1}^N a_{\textbf{w}_i}\ket{\textbf{w}_i}\right)^\top J^{\otimes 2k}\left(\sum_{j=1}^N b_{\textbf{w}_j}\ket{\textbf{w}_j}\right) \\
&= \left(\sum_{i=1}^N a_{\textbf{w}_i}\ket{\textbf{w}_i}\right)^\top\left(\sum_{j=1}^N b_{{\textbf{w}}_j} (-1)^{|\bar{\textbf{w}}_j|}\ket{\bar{\textbf{w}}_j}\right)\\
&=\sum_{\textbf{w}\in\{0,1\}^{2k}} (-1)^{|\textbf{w}|}a_{\textbf{w}}b_{\bar{\textbf{w}}}\;.
\end{align*}
Notice that $\ket{0}\otimes\ket{\textbf{w}}$ and $\ket{1}\otimes\ket{\textbf{w}}$ differ in one tensor factor, whereas the Hamming distance between points in $\mathcal{E}$ is even. Then one of $a_\textbf{w}$ and $b_\textbf{w}$ must be 0, hence $\textbf{a}^\top\textbf{b}=\sum a_\textbf{w}b_\textbf{w}=0.$ Since $\textbf{w}$ has $2k$ entries, $|\textbf{w}|$ and $|\bar{\textbf{w}}|$ have the same parity. Thus, one of $a_\textbf{w}$ and $b_{\bar{\textbf{w}}}$ is 0 and we similarly have $\textbf{a}^\top J^{\otimes 2k}\textbf{b}=\sum (-1)^{|\textbf{w}|}a_{\textbf{w}}b_{\bar{\textbf{w}}}=0$. This shows that $\Phi_{(1)}^\top  J^{\otimes 2k} \Phi_{(1)}^\top$ and $\Phi_{(1)}^\top \Phi_{(1)}^\top$ are diagonal. Note that, if $\textbf{w}\neq \textbf{0}$, then $a_\textbf{w}$ equals $1$ if $|\textbf{w}|$ is even and $0$ otherwise. Thus we compute
\begin{align*}
\textbf{a}^\top J^{\otimes 2k}\textbf{a}&=\sum (-1)^{|\textbf{w}|} a_\textbf{w} a_{\bar{\textbf{w}}} \\
&=\sum a_\textbf{w} a_{\bar{\textbf{w}}} \\
&=2(a_{\textbf{0}}a_{\bar{\textbf{0}}})+\sum_{\textbf{w}\in\{0,1\}^{2k}\setminus\{\textbf{0},\bar{\textbf{0}}\}}a_\textbf{w} a_{\bar{\textbf{w}}}\\
&= 2(1-2^{2k-1})+(2^{2k-1}-2).
\end{align*}
Similarly, $b_\textbf{w}$ equals $1$ if $|\textbf{w}|$ is odd and $0$ otherwise, so
\begin{align*}
    \textbf{b}^\top J^{\otimes 2k}\textbf{b}&=\sum (-1)^{|\textbf{w}|} b_\textbf{w} b_{\bar{\textbf{w}}}=-2^{2k-1}.
\end{align*}
It follows that $\textbf{a}^\top J^{\otimes 2k}\textbf{a}=\textbf{b}^\top J^{\otimes 2k}\textbf{b}$ so that $\Omega_{(1)}^\top J^{\otimes 2k} \Omega_{(1)}^\top$ is a multiple of the identity. Additionally, we have
\[
\textbf{a}^\top\textbf{a}=\sum a_\textbf{w}^2=(2^{2k-1}-1)^2+(2^{2k-1}-1) ,\quad\text{and}\quad
\textbf{b}^\top\textbf{b}=2^{2k-1}.
\]
which gives $\textbf{a}^\top\textbf{a}-\textbf{b}^\top\textbf{b}=(2^{2k-1}-1)^2 -1$. Thus $\textbf{a}^\top\textbf{a}\neq\textbf{b}^\top\textbf{b}$ so that $\Phi_{(1)}^\top  \Phi_{(1)}^\top$ has distinct eigenvalues when $k\geq 2$ or $n=2k+1\geq 5$. We have equality $\textbf{a}^\top\textbf{a}=\textbf{b}^\top\textbf{b}$ when $k=1$ or $n=3$. Due to the symmetry of the tensor $\Phi$, the same calculations apply for flattenings $\Phi_{(i)}$ with $i>1$. This concludes the proof of \cref{lem:gen}.

\subsubsection{The 3-qubit case}\label{sec:3f} Using the methods we have developed, we now examine the orbit classification problem for the SLOCC group action on $\cH_3$. Suppose $\Phi\in\cH_3$ such that $\Phi_{(i)}(J\otimes J)\Phi_{(i)}^\top$ is invertible for $i=1,2,3$. By \cref{lem:preodd}, there exists $\Psi$ in the SLOCC orbit of $\Phi$ such that $\Psi_{(i)}(J\otimes J)\Psi_{(i)}^\top$ is a multiple of the identity $\forall i$. Running the first two steps of \cref{alg:slocc-odd} on a computer with randomly selected $\Phi$, we find that the output $\Psi=av_1+bv_2$ is always a complex linear combination of the tensors
\[
v_1=\ket{001}+\ket{010}+\ket{100}-\ket{111},\quad
v_2=\ket{101}+\ket{110}-\ket{000}+\ket{011}.
\]
A direct computation shows that 
    \[
    U(z_1)\otimes U(z_2)\otimes U(z_3)
    (av_1+bv_2)=\begin{pmatrix} v_1 & v_2\end{pmatrix}U(z_1+z_2+z_3)\begin{pmatrix} a \\ b\end{pmatrix},
    \]
    where $U(z)=\begin{pmatrix}
        \cos z & -\sin z \\
        \sin z & \cos z
    \end{pmatrix}$ for $z\in\CC$. Therefore the $\SO_2^{\times 3}$-representation $\cH_3$ contains a subrepresentation
    $
\SO_2^{\times 3}\to\SO_2\to\GL(\{v_1,v_2\}).
    $ Since $\SO_2$ is abelian, the subrepresentation splits further into two 1-dimensional subrepresentations which correspond to the simultaneous eigenvectors $w_1=v_1+\im v_2$ and $w_2=v_1-\im v_2$ associated to characters $e^{-\im z}$ and $e^{\im z}$ respectively, i.e.
\[
U(z)\begin{pmatrix} 1 \\ \im \end{pmatrix}=e^{-\im z}\begin{pmatrix} 1 \\ \im \end{pmatrix}\quad\text{and}\quad U(z)\begin{pmatrix} 1 \\ -\im \end{pmatrix}=e^{\im z}\begin{pmatrix} 1 \\ -\im \end{pmatrix},\quad\forall z\in\mathbb{C}.
\]  
   Given $\alpha,\beta\in\CC$ we can set $z=(\arg \alpha+\arg \beta)/2+\im \,\ln\sqrt{|\alpha\beta|}$ so that 
   \[e^{-\im z} \alpha = e^{\im z}\beta=\sqrt{|\alpha\beta|}\exp(\im (\arg\alpha + \arg\beta)/2).
   \] Then $I_2\otimes I_2\otimes U(z)(\alpha w_1+\beta w_2)$ is a scalar multiple of $v_1$. Hence we obtain normal forms for almost all tensors constituting the line
   \[
   av_1 = a(\ket{001}+\ket{010}+\ket{100}-\ket{111}),\quad a\in \CC.
   \]
The above is true for the specific case when $\Phi=\frac{1}{\sqrt{2}}(\ket{000}+\ket{111})$ is the GHZ state. Thus we recover the fact that almost all states in $\mathbb{P}\cH_3$ are SLOCC equivalent to the GHZ state. This normal form for the GHZ state also appears in the Freudenthal triple system classification from \cite{Borsten_2013}.

\begin{table}[h]
    \begin{center}
    \bgroup
\def\arraystretch{1.1}
\begin{tabular}{|c||c|c|c|c|c|c|}
    Normal form of $\Phi$ & rank $\pi_1(\Phi)$ & rank $\pi_2(\Phi)$ & rank $\pi_3(\Phi)$\\
    \hline
    $\ket{000}+\ket{111}$ & 2 & 2 & 2\\
    $\ket{001}+\ket{010}+\ket{100}$ & 1 & 1 & 1 \\
    $\ket{001} + \ket{111}$ & 0 & 0 & 1 \\
    $\ket{010} + \ket{111}$ & 0 & 1 & 0 \\
    $\ket{100} + \ket{111}$ & 1 & 0 & 0 \\
    $\ket{000}$ & 0 & 0 & 0 
\end{tabular}
\egroup
\end{center}
    \caption{Ranks of the reduction maps for 3 qubits.}
    \label{fig:tracing3qubits}
\end{table}

It is known that there are six SLOCC equivalence classes in $\PP \cH_3$ with normal forms listed in \cref{fig:tracing3qubits}. We find that the ranks of the matrices $\pi_i(\Phi)=\Phi_{(i)}(J\otimes J)\Phi_{(i)}^\top$ for $i=1,2,3$ is enough to distinguish SLOCC orbits in $\mathbb{P}\cH_3$. It is known that ranks of flattenings and the hyperdeterminant also distinguish these orbits. The multilinear ranks separate all but the top two orbits. The top two orbits both have multilinear rank (2,2,2) and the (non)-vanishing of the hyperdeterminant separates these two (see \cite[Ex.~4.5, p.~478]{GKZ}). 

\section*{Acknowledgements}
The authors recognize partial support from the Thomas Jefferson Foundation, CNRS, and UTBM. The authors would also like to thank Frédéric Holweck and Nick Vannieuwenhoven for useful discussions. We thank the anonymous referees who provided useful feedback.

\newcommand{\arxiv}[1]{\href{http://arxiv.org/abs/#1}{{\tt arXiv:#1}}}
 \bibliographystyle{siamplain}
\bibliography{ref}
 \end{document}